\documentclass[journal,12pt,onecolumn,doublespace,draftcls]{IEEEtran}

\usepackage[cmex10]{amsmath}
\usepackage{amssymb}
\usepackage{cite}
\usepackage{amsmath}
\usepackage{wasysym}
\usepackage{booktabs}
\usepackage{paralist}
\usepackage{balance}
\usepackage{subfigure}
\usepackage{multicol}
\usepackage{graphicx}
\usepackage{epsfig}
\usepackage{epstopdf}
\usepackage[ruled,linesnumbered]{algorithm2e}
\usepackage{bm}
\usepackage{caption}
\usepackage{amsmath, amssymb,amsthm}
\usepackage{upgreek}
\usepackage{CJK}
\usepackage{color}
\usepackage{float}
\usepackage{xfrac}

\newtheorem{proposition}{Proposition}
\newtheorem{lemma}{Lemma}
\newtheorem{corollary}{Corollary}

\newtheorem{remark}{Remark}

\ifCLASSINFOpdf
\graphicspath{{pic/}}
\else
\fi
\begin{document}
%

\title{Modeling and Performance Analysis  in Cache-enabled  Millimeter Wave HetNets with Access and Backhaul Integration}
%

\author{Hao~Wu,
        Chenwu~Zhang,
        Hancheng~Lu,~\IEEEmembership{Member,~IEEE},
        Qi~Hu
\thanks{Hao Wu,  Chenwu~Zhang and  Hancheng Lu are  with CAS Key Laboratory of Wireless Optical Communication, University of Science of China, Hefei 230027, China. (Email: hwu2014@mail.ustc.edu.cn,  cwzhang@mail.ustc.edu.cn, hclu@ustc.edu.cn).

Qi~Hu is with University of Science and Technology of China, Hefei 230027, China. (Email:hq1998@mail.ustc.edu.cn ).

}
}

\maketitle


%
\IEEEpeerreviewmaketitle

\section{INTRODUCTION}
Recently,   mmWave-based access and backhaul integration heterogeneous cellular networks (mABHetNets) has been envisioned in 5G dense cellular networks to satisfy the rapidly growing traffic demand \cite{IAB1,IAB2,IAB3,3GPPIAB}. In the mABHetNets, high-power mmWave MBSs are overlaid by denser  lower-power mmWave SBSs where   MBSs and SBSs provide high  rate service to the users by wireless access link while the MBS   maintains the backhaul capacity of the SBSs by the wireless backhaul link. Both the access link and the backhaul link   share  the same mmWave spectrum resources, which is called a mmWave-based access and backhaul integration architecture. 
Average potential throughput (APT) and area spectral efficiency (ASE) have become the two major performance metrics for 5G dense cellular networks\cite{Throuhgput1,Throuhgput2,Throuhgput3,ASE0,ASE1,ASE2}. APT focuses on analyzing   user's average throughput with the specific rate requirement \cite{Throuhgput1}. \cite{Throuhgput2} analyze  APT of  different user's SINR requirements in a new  path loss model. Further, \cite{Throuhgput3} investigates APT of user    in line of-sight (LOS) and non-line-of-sight (NLOS) scenarios respectively.  Besides, another widely used metric ASE is defined as  the spectral efficiency  per unit area of cellular networks\cite{ASE0}.   \cite{ASE1} discussed the ASE in the different SBS antenna gain patterns. \cite{ASE2} analyzes ASE under different user association strategies in D2D millimeter-wave networks.

In mABHetNets, mmWave spectrum bandwidth partition between access link and backhaul link has played an important role in APT and ASE\cite{Partition1,Partition2,Partition3}. Since partitioned spectrum bandwidth between the access and backhaul is orthogonal, the interference between the access and backhaul   is avoided and the wireless   rate is improved. As the user's dare rate is influenced by both access link rate and backhaul link rate, \cite{Partition1}  explores the optimal partition of access and backhaul bandwidth to maximize the rate coverage. In a mmWave unified access and backhaul  network, \cite{Partition2} leverages allocated resource ratio between radio access and backhaul to study maximization of network capacity by considering the fairness among SBSs. \cite{Partition3} jointly studies the beamforming and bandwidth partition to improve the network capacity of mABHetNets. However, since the backhaul link may suffer relatively high  path loss, a large amount of mmWave spectrum is     occupied by the backhaul link to maintain the backhaul link capacity. According to the findings in \cite{Partition3}, up to 50\% mmWave spectrum will be used in backhaul link to satisfy the high speed data traffic.  Such stubborn ``\emph{spectrum occupation}'' in mABHetNets has  restricted network performance such as ASE and APT to achieve a better possible improvement.

Nowadays, enabling caching at the wireless edge such as SBSs called cache-enabled mABHetNets has been considered as a promising way to improve the network performance\cite{2u,caching,caching1,MostPop}. Statistical reports have shown that a few popular files requested by many users account for most of the backhaul traffic load \cite{2u,caching}. Based on this fact, equipping caches at all BSs for caching the most popular contents becomes an effective method to offload the data traffic of backhaul \cite{MostPop}. In detail, popular files can be proactively cached at SBSs during off-peak time, and delivered to users when requested, which can significantly alleviate backhaul traffic pressure. By exploiting these benefits, equipping caches in the mABHetNets brings an opportunity to overcome \emph{spectrum occupation} problem. When the backhaul traffic is offloaded by caching popular files at the cache of SBSs, the part of mmWave spectrum in backhaul can be transferred to the access link. With more cache capacity, more spectrum is used for the access link, both   ASE and APT can be increased.

However, as we know that, additional caches will consume the  power of   BSs\cite{BSPowerModel,BSPowerModel1}. Both  \cite{BSPowerModel1} and \cite{BSPowerModel}  think the caching energy consumption  is the key part of the total power  consumption. Since   BSs has a limited energy resource, when   the caching power consumption is introduced,   the  transmission power is recduced, which will lower the data rate.  Therefore, in this paper, we attempt to explore the impact of caches on the network performance of cache-enabled mABHetNets, especially in terms of APT and ASE. We also  identify the impacts of  some other cache-related factors. We first derive the basic  SINR distriution of the mABHetNets. Then, we further study the APT and ASE in the cache-enabled mABHetNets. To the best of our knowledge, there is no theoretical research that investigates the performance of the mABHetNets when caches are involved. Motivated by such fact, we carry out the analytical study for the cache-enabled  mABHetNets.  The major contributions of this paper are summarized as follows.
\begin{itemize}
\item We develop a tractable analytical framework by stochastic geometry tool to study the cache-enabled mABHetNets. Considering the LoS and NLoS transmission in mmWave, we   derive the expression of  APT and ASE where the key factors (e.g. cache capacity, bandwidth partition ratio etc. ) are identified.
\item With the derived performance expressions, we analyze APT and ASE over  the cache capacity and bandwidth partition. In order to make the ASE expressions more tractable, we derive ASE in the special case(e.g., the noise-limited and the interference-limited case).
\item We find that there exists case the optimal cache capacity. With the  small cache capacity,  and the corresponding APT and ASE are improved. However, when the cache capacity is larger, both APT and ASE will be decreased. With the optimal cache capacity, the optimal usage ratio of  mmWave spectrum  for access   can be improved from 0.5 to 0.8. Besides,  we also  see  that, some other caching parameters  have a significant influence on the network performance.
\end{itemize}
Finally, extensive numerical and simulations results are carried out to validate the motivation and effectiveness of our work.

The rest of the paper is organized as follows: Sect. II gives an overview of the system model. The SINR distribution of cache-enabled mABHetNets   are derived  in Sec. III. Then, APT and ASE are further analyzed in Sec. IV, respectively. Last, the numerical  results are presented in Section V. Finally, we conclude the paper in Section VI.

\section{SYSTEM MODEL}
In this section, we consider a downlink mABHetNets consisting  MBS tier and   SBS tier. By the  stochastic geometry tool, the location of the BSs and users are described. Then,  file caching model and the SINR model are illustrated. Last, the bandwidth partition between the access and backhaul link is introduced.

%

\begin{figure}[htbp]
  \centering
  \includegraphics[width=4.5in]{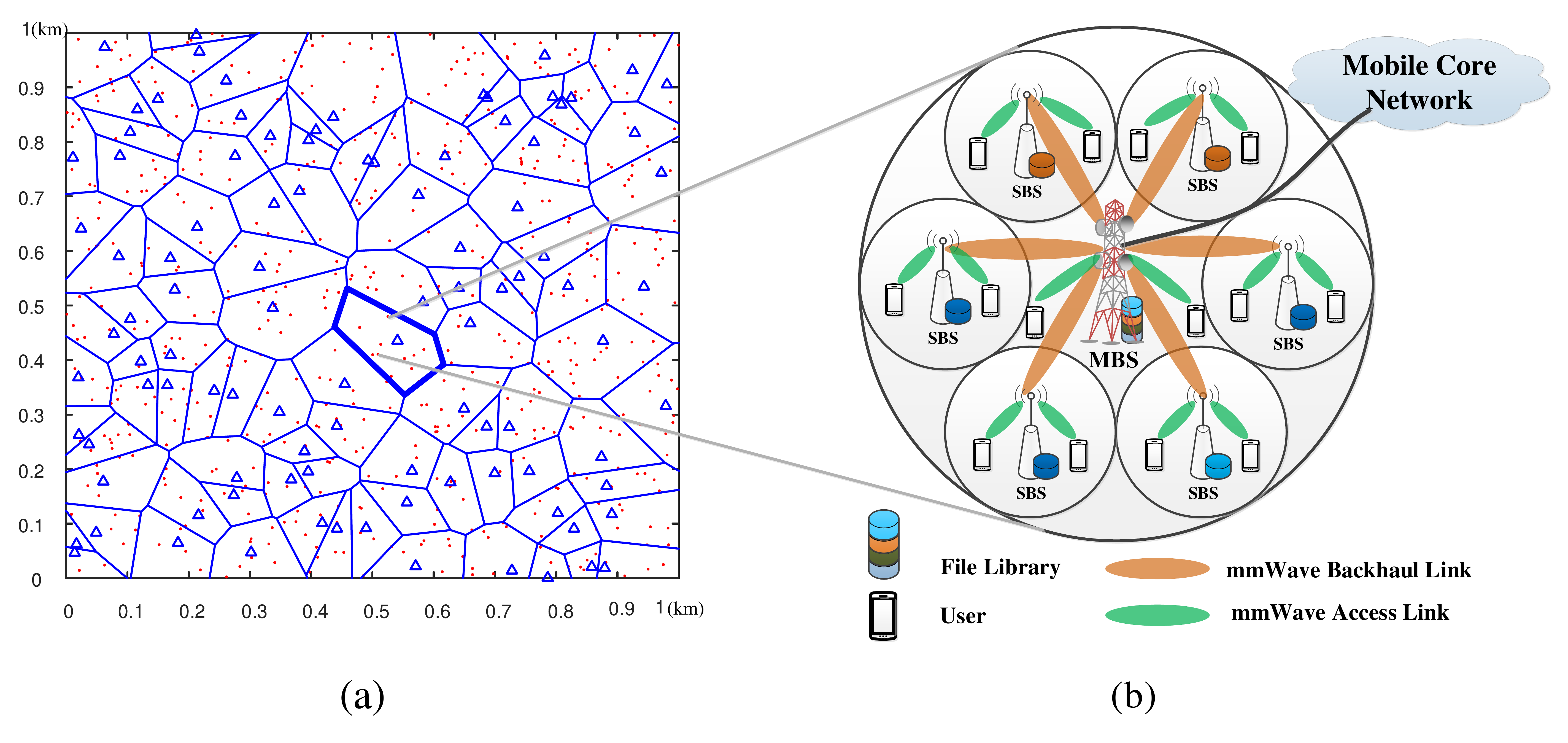}
  \captionsetup{font={small}}
  \caption{(a) Example of downlink mABHetNets with two tiers of BSs:   MBSs (blue triangle) are overlaid with denser SBSs (red point) (b) a megascopic cache-enabled mABHetNets}
  \label{example}
\end{figure}

\subsection{Network Model}
 In this paper, with the aid of stochastic geometry approach, a tractable analytical network model is proposed for characterizing the performance of our mABHetNets.  As it is shown in Fig. \ref{example}, downlink communication in a two-tier mABHetNets is considered. The first tier of mABHetNets consist of lower-power SBSs while the second tier consists of higher-power MBSs. The MBSs is connected to the mobile core network by high-capacity  optical fiber links. Besides, the MBSs provide wireless backhaul connections to the SBSs via   broad mmWave spectrum. The locations of the MBSs  and SBSs are assumed to follow independent Poisson point processes (PPPs), which are denoted by $\Phi_m\in \mathbb{R}_2$ and $\Phi_s\in \mathbb{R}_2$  with densities and $\lambda_m$ and $\lambda_s$, respectively. Both MBSs and SBSs provide the access service to the users, which is also modeled according to independent homogeneous PPPs $\Phi\in \mathbb{R}_2$ with density $\lambda$. $\lambda$
is assumed to be sufficiently larger than $\lambda_m$ and $\lambda_s$ so that each BS has
at least one associated UE in its coverage area. The access link is also using mmWave spectrum. The above
model where the nodes are distributed using PPPs has been shown to be quite effective for system-level performance evaluation of cellular networks \cite{PPP}.

\begin{table}
\centering\small
\caption{Main Symbols}
\label{Symbols}
\renewcommand\arraystretch{0.8}
\begin{tabular}{cp{12cm}}
\toprule
\textbf{Symbol}                     & \textbf{Meaning}                                       \\
\midrule
$W,W_{ac},W_{bh}$                         &Total spectrum bandwidth, access link bandwidth and backhaul link bandwidth\\
$\eta$                              &mmWave bandwidth partition ratio for access link\\
$\alpha_\mathrm{L},\alpha_{\mathrm{NL}}$                            & Path loss exponet in LoS and NLoS transmission\\
$\lambda_s,\lambda_m,\lambda$                             &Density of MBS, SBS and user, respectively\\
$C$                                 &Cache capacity of a SBS\\
$F$                                 &Number of files \\
$p_h$                               &Cache hit ratio of a SBS\\
$P_{m}^{tot},P_{s}^{tot}$           &Total power of an MBS and an SBS, respectively\\
$P_{m}^{tr},P_{s}^{tr}$             &Transmission power of an MBS and an SBS, respectively\\
$B_m,B_s$                           &The association bias factor of MBS and SBS\\
$w_{ca}$                            &Caching power consumption coefficient\\
\bottomrule
\end{tabular}
\end{table}
\subsection{Caching Model}
The file library is denoted by $\mathcal{F}$ and there are $|\mathcal{F}|= F$ files in the library. It is assumed that each file has the equal size \cite{Samesize}
for mathematical tractability and notational simplicity.  Note that when the file sizes are different from each other, such assumption can still hold  since those files can be divided into chunks of equal size. Different files have different popularities. The file popularity distribution changes with time slowly so that it can be regarded as static \cite{StableDistribution}. Then, based on  collected statistical information and  the  machine learning algorithm, the file popularity can often be predicted \cite{Popularity}.  Zipf distribution  is widely used to model the popularity of file $f$, $\forall f \in\{1,...,F\}$: the probability of  the $f$-th file is
$
  p_{f}=\frac{f^{-\gamma_p}}{\sum_{g=1}^{F}g^{-\gamma_p}},
$
where $\gamma_p$ is the Zipf exponent reflecting different levels of skewness of the distribution\cite{Zipf,Zipf1}. The typical value of $\gamma_p$ is between 0.5 and 1.0, where higher value causes more ``peakiness'' of the distribution \cite{Zipf1}.

For the cache-enabled MBS and SBS, the cache capacity of the SBS and MBS are denoted by $F$ and $C$ (file units), respectively. As each MBS is often equipped with  large  cache capacity, we assume that the MBS can cache the total file library \cite{MBSFileLibrary}. In this paper, the most popular caching strategy is applied for each BS : each BS caches the
most popular contents until its storage is full\cite{MostPop}. Hence, the cache hit ratio of a SBS can be calculated as
\begin{equation}\label{HitRatio}
  p_{h}=\frac{\sum_{f=1}^{C}f^{-\gamma_p}}{\sum_{g=1}^{F}g^{-\gamma_p}}.
\end{equation}

\subsection{Wireless Transmission Model}
The mmWave based  wireless   link   can be either line of sight (LOS) or non-line of sight (NLoS) transmission. Then, we consider the following path loss function defined in  \cite{LOSModel}:
\begin{align}\label{LOSModel}
 \operatorname{L}(r)
 =\left\{
 \begin{array}{ll}
 {A_{\mathrm{L}} r^{-\alpha_{\mathrm{L}}},}      &{\text{with LOS probability $\mathcal{P}_\mathrm{L}(r)$}} \\
 {A_{\mathrm{NL}} r^{-\alpha_{\mathrm{NL}}},}    &{\text{with NLoS probability $\mathcal{P}_\mathrm{NL}(r)=1-\mathcal{P}_\mathrm{L}(r)$}},
 \end{array}
 \right.
\end{align}
together with the   LoS probability  $\mathcal{P}_{\mathrm{L}}(r)=\min\left(\frac{18}{r},1\right)(1-e^{-\beta r})+e^{-\beta r}$
and $r$ is the transmission distance.   $\beta \geq 0$ is the parameter that captures
density and size of obstacles between the transmitter and the reveicer. As $\beta$ increases, the size and density of obstacles increase, which results in lower probability of LoS transmission. The propagation is always in LOS condition when $ r \leq 18 m$. In practice, this implies that, for denser mABHetNets, the probability of LOS coverage is very close to one and  some NLOS  transmission  could be neglected.

By involving the cache power consumption,   we give the power model   of  one MBS and one SBS as follows\cite{BSPowerModel}:$ \rho_m P_{m}^{tr}+P_{m}^{fc}+P_{m}^{ca}=\rho_{m} P_{m}^{\mathrm{tr}}+P_{m}^{\mathrm{fc}}+w_{\mathrm{ca}}F,     \rho_s P_{s}^{tr}+P_{s}^{fc}+P_{s}^{ca}=\rho_s P_{s}^{tr}+P_{\mathrm{s}}^{\mathrm{fc}}+w_{\mathrm{ca}}C$
where $P_{m}^{tr}$ and $P_{s}^{tr}$ denote transmit powers consumed at a MBS and a SBS, respectively. $\rho_m(\rho_s)$ reflects the impact of power amplifier and cooling on transmit power. $P_{m}^{fc}(P_{s}^{\mathrm{fc}})$ is the fixed circuits-related power consumption at a BS. $P_{m}^{ca}(P_{s}^{ca})$ is the power consumption for caching files at a BS. Usually, $P_{m}^{fc}(P_{s}^{fc})$ is assumed to be a fixed  power consumption constant\cite{CircuitPowerModel}. To quantify  power consumption for caching, we adopt a power-proportional model, which is widely used in cache enabled radio access network  \cite{BSPowerModel,BSPowerModel1}. In the power-proportional model, the caching
power consumption is proportional to the cache capacity. Then, the caching power  consumption of the MBS and SBS are given as follows: $P_{m}^{ca}=w_{ca}F,~P_{s}^{ca}=w_{ca}C$, where $\omega_{ca}$ is the power coefficient of cache hardware in watt/bit. $F$ is the total size of all the files. Here, we assume that the cache storage of the MBS is large enough and the library of all files are cached in MBS \cite{MBSFileLibrary}.  $C$ is the cache storage of a SBS in file unit.

For each SBS, since the total power consumption $P_{s}^{tot}$ is usually given, we can get $P_s^{tr}=\frac{P_{s}^{tot}-P_{\mathrm{s}}^{\mathrm{fc}}-w_{\mathrm{ca}}C}{\rho_s}$ $
=P'_s-{w'}_{\mathrm{ca}}^{s}C$, where $P'_s=\frac{P_{s}^{tot}-P_{\mathrm{s}}^{\mathrm{fc}}}{\rho_s}$ and ${w'}_{ca}^s=\frac{w_{ca}}{\rho_s}$. Here, we assume that the MBS has a large cache capacity and contains all the files in  the file library. The total power consumption of a MBS is $P_{m}^{tot}$. Then the $P_m^{tr}=\frac{P_{m}^{tot}-P_{\mathrm{m}}^{\mathrm{fc}}-w_{\mathrm{ca}}C}{\rho_m}
=P'_m-{w'}_{\mathrm{ca}}^{m}F$ is fixed where $P'_m=\frac{P_{m}^{tot}-P_{\mathrm{m}}^{\mathrm{fc}}}{\rho_m}$ and ${w'}_{ca}^m=\frac{w_{ca}}{\rho_m}$. Note that considering the total  power  consumption,  the maximum cache capacity  is  $\frac{P_{s}^{tot}}{w_{\mathrm{ca}}}$.

The analysis in this paper is done for the user located at the origin referred to as the typical user. Therefore, the signal-to-interference-plus-noise ratio (SINR) of a typical user at a random distance $r$ from its associated SBS and MBS   are
\begin{small}
\begin{align}\label{SBSUSERSINR}
\operatorname{SINR}_{s}(r)
&=\frac{P_{s}^{tr} B_{s}h_{s  } L(r)}{I_s+I_m+N_{0}}
 =\frac{(P'_s-w'_{\mathrm{ca}}C) B_{s}h_{s  } L(r_{s })}
   {\sum\limits_{i \in \Phi_{s} \backslash b_{s,0}}(P'_s-w'_{\mathrm{ca}}C) B_{s } h_{s, i}L(r_{s, i })
   +\sum\limits_{l \in \Phi_{m} } P_{m}^{tr}  B_{m}h_{m, l }L(r_{m, l })+N_{0}}
\end{align}
\begin{align}\label{MBSUSERSINR}
\operatorname{SINR}_{m}(r)
&=\frac{P_{m}^{tr} B_{m}h_m L(r_{m} )}
{I'_s+I'_m+N_{0}}=\frac{P_{m}^{tr} B_{m}h_m L(r_{m} )}
{\sum\limits_{i \in \Phi_{s}}  (P'_s-w'_{\mathrm{ca}}C)B_{s} h_{s, i} L(r_{s, i})+\sum\limits_{l \in \Phi_{m} \backslash b_{m,0}} P_{m}^{tr} B_{m} h_{m, l} L(r_{m, l})+N_{0}}
\end{align}
\end{small}
where   $B_s,B_m$ are the association bias factor of SBS and MBS.   $h_{s},h_{m }$ are the small-scale fadings from SBS and MBS. $L(r_{m}),L(r_{s})$ are the path losses   from the serving SBS or MBS to the typical user. $r_{s}$ ($r_{m}$) is the distance between the assocation SBS $b_{s,0}$  (assocation MBS $b_{m,0}$) and the typical user.  $ r_{s, i}$($ r_{m,l}$) is the distance between the $i$-th SBS($l$-th MBS )and the typical user.    $N_{0}$ is   the additive white Gaussian noise component.

To backhaul the data traffic of the SBSs, the MBS provides the wireless backhaul link. For a typical SBS that a random distance $r_{bh}$ from its associated MBS,  the SINR of the signal from the MBS to the SBS in downlink backhaul is then given as,
\begin{align}\label{M-SBSSINR}
\operatorname{SINR}_{bh} (r_{bh})
&=\frac{P_{m}^{tr} B_m h_{m } L(r_{bh})}
{ I_{bh}+  N_0}=\frac{P_{m}^{tr}B_m h_{m } L(r_{bh})}
{ \sum_{i \in \Phi_{m} \backslash b_{m,0 }} P_{m}^{tr} B_m  h_{m,i}L(r_{bh,i})+  N_0}
\end{align}

\subsection{Bandwidth Partition Model}
\begin{figure}[htbp]
  \centering
  \includegraphics[width=4.5in]{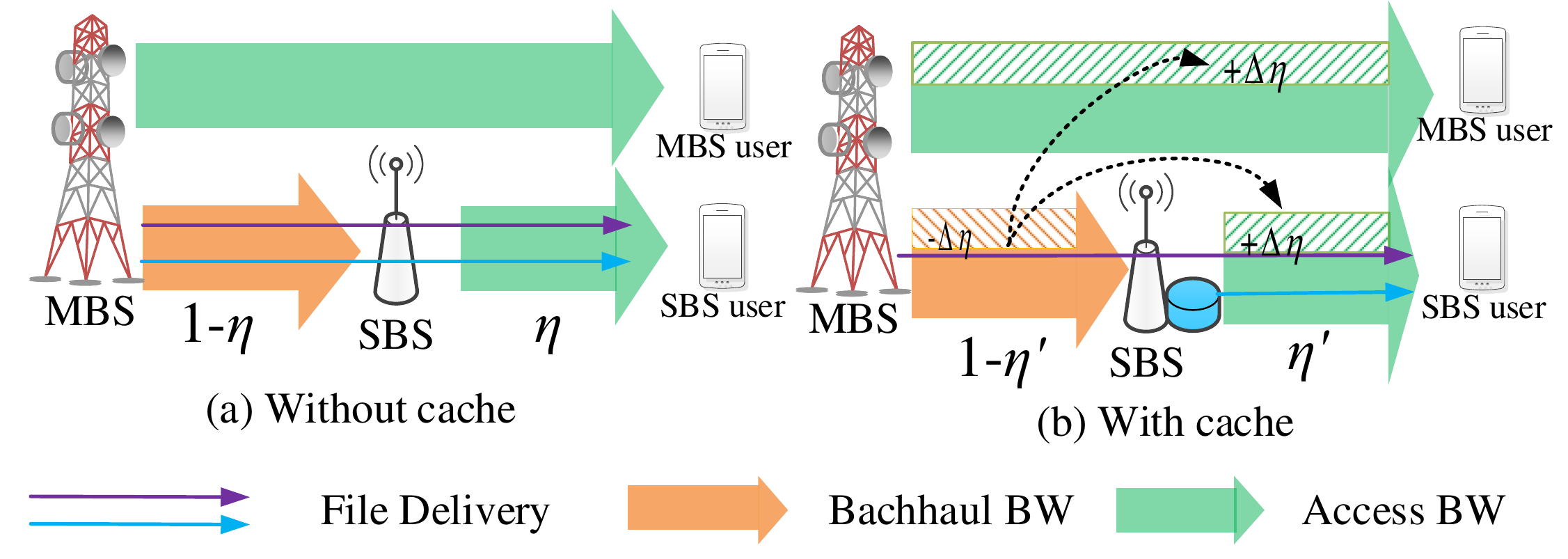}
  \captionsetup{font={small}}
  \caption{Bandwidth partition between access and backhaul (a) original and (b)with cache}
  \label{schematic}
\end{figure}
In this paper, we focus on the downlink data transmission using the mmWave spectrum. In Fig.\ref{schematic}, both the access link and backhaul link use the mmWave spectrum. To avoid the interference between the access and backhaul  link, the access and backhaul use orthogonal spectrum resources. The total mmWave bandwidth $W$ for downlink transmission is partitioned into two parts: $W_{ac}=\eta W$  for access and $W_{bh}=(1-\eta )W$ for backhaul.  $\eta \in [0, 1]$  is access-backhaul bandwidth partition ratio and denotes the    part of spectrum for access link. The file delivery is both related with the access link and backhaul link. Then  by changing the   partition ratio,  both the access rate and backhaul rate in the same transmission path can maintain an effective  transmission. When the cache is introduced, cached files can directly delivered from SBS and the backhaul traffic can be saved. Thus, part mmWave spectrum resource in backhaul link can be shifted to access link and the data rate is improved.

\section{SINR Distribution   of mABHetNets}

In this section, we derive the expression of SINR distribution of the typical user conditioned on its association selections  and later decondition over them. As users may be covered either by the SBS or MBS,  we first derive the PDF of the distance   between the user and  the serving SBS and MBS. Further,  SINR distributions of the users associated with the serving SBS  and the serving MBS are obtained. Besides, the SINR coverage probability of the SBS covered by the MBS is also obtained.

\subsection{The PDF of Distance to Nearest Base Station}
First of all, we need to derive the probability distribution function (PDF) of the distance between the typical  user and its nearest BS. We focus on the  typical user at the origin.  When the typical user communicates with the closest BS at a distance $r$, no other BS can be closer than $r$. In other words, all interfering BSs must be farther than $r$. Since the typical user is associated with the closest BS via either LoS or NLoS channel, we derive these PDFs in following Lemma, respectively
\begin{lemma}\label{LemmaUserClosestSBSProbabilty}
The PDF of $r$ (the distance between the typical user and the  nearest SBS or MBS via a LoS/NLoS path) is written as
    \begin{align}\label{LOSUserClosestSBSProbabilty}
        f_{R_{k}}^{\mathrm{L}}(r)     &=\mathcal{P}_L(r) \times \exp \left(-\pi r^{2} \lambda_k\right) \times 2 \pi r \lambda_k, \\
        f_{R_{k}}^{\mathrm{N L}}(r)   &= \mathcal{P}_\mathrm{NL}(r) \times \exp \left(-\pi r^{2} \lambda_k\right) \times 2 \pi r \lambda_k,
    \end{align}
where $k\in \{s,m\}$ denotes the index of SBS tier or MBS tier.
\end{lemma}
\begin{proof}
The detailed proof procedure can be found in Appendix \ref{Appe1:NearDistance}.
\end{proof}

\begin{remark}
  The uncached file data will be delivered from the MBS to the SBS by the mmWave based wireless backhaul. Similar to the above analysis in Lemma \ref{LemmaUserClosestSBSProbabilty}, the PDFs of distance $r$ (between  the SBS  and the  associated the nearest MBS via a LoS/NLoS path)  are
\begin{align}\label{PDFSBSMBSLOS}
  f_{R_{bh}}^{\mathrm{L}}(r)&=\mathcal{P}_L(r) \times  \exp \left(-\pi r^{2} \lambda_m\right) \times 2 \pi r \lambda_m \\
  f_{R_{bh}}^{\mathrm{NL}}(r)&=\mathcal{P}_{\mathrm{NL}}(r) \times  \exp \left(-\pi r^{2} \lambda_m\right) \times 2 \pi r \lambda_m
\end{align}
\end{remark}

\subsection{Association Probability}
In mABHetNets, due to the different  transmission path (i.e., LoS and NLoS) and the densities of the SBS   and MBS, we need to analyze the probability that a user  is associated with SBS tier or with MBS tier.  Besides, since in the SBS backhaul association, SBS may be associated with MBS via different  transmission paths, different SBS backhaul association probabilities should be derived.
\subsubsection{ user association probability}
We consider a user association based on maximum biased received power, where a mobile user is associated with the strongest BS in terms of  the received power at the user. Then, for   a typical user associated with the SBS tier via LoS path and NLoS path, the received powers  are $P_{s}^{tr}  B_sh_{s}  A_\mathrm{L} r^{-\alpha_\mathrm{L}}$ and $P_{s}^{tr}  B_sh_{s}  A_\mathrm{NL} r^{-\alpha_\mathrm{NL}}$, respectively. For   a typical user associated with the MBS tier via LoS path and NLoS path, the received powers  are $P_{m}^{tr}  B_mh_{m}  A_\mathrm{L} r^{-\alpha_\mathrm{L}}$ and $P_{m}^{tr}  B_mh_{m}  A_\mathrm{NL} r^{-\alpha_\mathrm{NL}}$, respectively.

Based on the maximum biased received power association strategy, each tier’s BS density and transmit power as well as transmission path determine the probability that a typical user is associated with a tier. The following lemma provides the per-tier association probability via LoS and NLoS path respectively.

\begin{lemma}\label{LemmaUserAssociatedSBS}
   For the given distance $r$, the probabilities that a typical user is associated with the SBS tier by LoS link and NLoS link are
\begin{align}\label{UserAssociatedSBS}
    F_{s}^{\mathrm{L}}(r) &=p_{ln}^{ss}(r)p_{ll}^{sm}(r)p_{ln}^{sm}(r)f_{R_{s}}^{\mathrm{L}}(r),\\
    F_{s}^{\mathrm{NL}}(r)&=p_{nl}^{ss}(r)p_{nl}^{sm}(r)p_{nn}^{sm}(r)f_{R_{s}}^{\mathrm{NL}}(r)
\end{align}
Then, the probabilities that a typical user is associated with the MBS tier by LoS link and NLoS link are
\begin{align}\label{UserAssociatedMBS}
  F_{m}^{\mathrm{L}}(r)&=p_{ln}^{mm}(r)p_{ll}^{ms}(r)p_{ln}^{ms}(r)f_{R_{m}}^{\mathrm{L}}(r),\\
  F_{m}^{\mathrm{NL}}(r)&=p_{nl}^{mm}(r)p_{nl}^{ms}(r)p_{nn}^{ms}(r)f_{R_{m}}^{\mathrm{NL}}(r)
\end{align}
where  $p_{ln}^{ss}(r)$ denotes the probability of the event that the user  obtains the desired LoS signal from SBS tier and the NLoS  interference from the SBS tier. The other probabilities have the similar definitions and can be found in the proof.
\end{lemma}
\begin{proof}
The detailed proof procedure can be found in Appendix \ref{Appen2:AssoPro}.
\end{proof}

\subsubsection{SBS backhaul association probability}
The SBS will be associated with the MBS by the wireless backhaul link. The SBS backhaul association strategy is also based on the maximum biased received power from the MBS. Since the backhaul transmission includes LoS link and NLoS link, there exist two   backhaul association probabilities. Similar to Lemma \ref{LemmaUserAssociatedSBS}, the probabilities are as follows.
\begin{remark}
Similar to Lemma \ref{LemmaUserAssociatedSBS}, the probabilities that a typical SBS is associated with   the MBS tier by LoS link and NLoS
link are
\begin{align}\label{SBSAssociatedMBS}
&F_{bh}^{L}(r)=p_{ln}^{bh}(r) f_{R_bh}^{\mathrm{L}}(r),\\
&F_{bh}^{NL}(r)=p_{nl}^{bh}(r) f_{R_bh}^{\mathrm{NL}}(r),
\end{align}
where $p_{ln}^{bh}(r)$ is the probability that the SBS is associated with the LoS MBS and
the interference is from NLoS MBS.  $p_{nl}^{bh}(r)$ is the probability that the SBS is associated with the NLoS MBS
and the interference is from LoS MBS. $p_{ln}^{bh}(r)$ and $p_{nl}^{bh}(r)$  can be found in Appendix \ref{Appen2:AssoPro}.
\end{remark}

\subsection{SINR Distribution}
To study the APT and ASE performance of mABHetNets, we need first  investigate the SINR distribution of the user  covered by SBS/MBS tier via access link or the SINR distribution of  the SBS   covered by MBS via backhaul link.
This SINR distribution is defined as the SINR coverage probability that the received SINR is above a pre-designated threshold $\gamma$:
\begin{equation}
P^{\operatorname{cov}}(  \gamma)=\operatorname{Pr}[\operatorname{SINR}>\gamma]
\end{equation}
Since the user is covered either by SBS tier or by MBS tier, we first give the two SINR distributions. Then we give the SINR distribution of the typical SBS when it is covered by MBS.
\begin{proposition}\label{ProUserSINRCoverageProbability}
1) SINR coverage probabilities of user covered by SBS tier and MBS tier:

The SINR coverage probability that the user is associated with the SBS or MBS is
\begin{align}\label{UserSBSSINRCoverageProbability}
&\mathbb{P}_{k}^{cov}(\gamma)=P_{k,L}^{cov}(\gamma)+P_{k,NL}^{cov}(\gamma)\\
 P_{k,L}^{cov}(\gamma)&=\int_{0}^{\infty}  \exp \left(\frac{-\gamma N_{0}}{P_{k}^{tr}B_kA_{\mathrm{L}} r^{-\alpha_{\mathrm{L}}} }\right)  \mathcal{L}_{I_{k}}^{\mathrm{L}}  F_{k}^{\mathrm{L}}(r)\mathrm{d}r\label{UserSBSSINRCoverageProbabilityLoS}\\
P_{k,NL}^{cov}(\gamma)&=\int_{0}^{\infty}  \exp \left(\frac{-\gamma N_{0}}{P_{k}^{tr}B_kA_{\mathrm{NL}} r^{-\alpha_{\mathrm{NL}}} }\right)  \mathcal{L}_{I_{k}}^{\mathrm{NL}}   F_{k}^{\mathrm{NL}}(r)d r\label{UserSBSSINRCoverageProbabilityNLoS}
\end{align}
where $k\in\{s,m\}$ denotes SBS or MBS, respectively.  $P_{k,L}^{cov}(\gamma)=\mathbb{E}_{r }\left[\mathbb{P}\left[\mathrm{SINR}_{k}^{\mathrm{L}} (r ) \geq \gamma\right]\right]$ is the probability that the user is covered by SBS or MBS  with LoS based signal
and $P_{k,NL}^{cov}(\gamma)=\mathbb{E}_{r }\left[\mathbb{P}\left[\mathrm{SINR}_{k}^{\mathrm{NL}} (r ) \geq \gamma\right]\right]$ is the probability that the user is covered by SBS  or MBS  with NLoS based signal. $\gamma$ is the threshold for successful demodulation and decoding at the receiver. Besides, $\mathcal{L}_{I_{s}}^{\mathrm{L}} = \mathcal{L}_{I_{s,m}}^{\mathrm{L}}\left(\gamma r^{\alpha_L}\right)$,$\mathcal{L}_{I_{s}}^{\mathrm{NL}} =\mathcal{L}_{I_{s,m}}^{\mathrm{NL}}\left(\gamma r^{\alpha_{NL}}\right)$, $\mathcal{L}_{I_{m}}^{\mathrm{L}}= \mathcal{L}_{I'_{s,m}}^{\mathrm{L}}\left( \gamma r^{\alpha_L}\right)$ and $\mathcal{L}_{I_{m}}^{\mathrm{NL}}= \mathcal{L}_{I'_{s,m}}^{\mathrm{NL}}\left( \gamma r^{\alpha_L}\right)$

2) SINR coverage probabilities of SBS covered by  MBS:

The SINR coverage probability that the SBS is covered by the  MBS via wireless backhaul is:
\begin{align}\label{BHSINRCoverageProbability}
  \mathbb{P}_{bh}^{cov}(\gamma) &  =P_{bh,\mathrm{L}}^{cov}(\gamma)+P_{bh,\mathrm{NL}}^{cov}(\gamma)
\end{align}
\begin{align}\label{BHSINRCoverageProbabilityLoS}
  P_{bh,\mathrm{L}}^{cov}(\gamma)=\int_{0}^{\infty} \exp \left(\frac{-\gamma N_{0}}{P_{m}^{tr}  B_m A_{\mathrm{L}} r^{-\alpha_{\mathrm{L}}} }\right)  \mathcal{L}_{I_{bh}}^{\mathrm{L}}\left(\gamma r^{-\alpha_L}\right) F_{bh}^{\mathrm{L}}(r)dr
\end{align}
\begin{align}\label{BHSINRCoverageProbabilityNLoS}
  P_{bh,\mathrm{NL}}^{cov}(\gamma)=\int_{0}^{\infty}\exp \left(\frac{-\gamma N_{0}}{P_{m}^{tr}  B_mA_{\mathrm{NL}} r^{-\alpha_{\mathrm{L}}} }\right)  \mathcal{L}_{I_{bh}}^{\mathrm{NL}}\left(\gamma r^{-\alpha_{NL}}\right) F_{bh}^{\mathrm{NL}}(r) d r
\end{align}
Note that,  considering a more general fading model such as Nakagami does not provide any additional design insights, but it does complicate the analysis significantly. Similar to \cite{Rayleigh} in our paper,   the special case of Nakagami-Rayleigh fading is considered .
\end{proposition}
\begin{proof}
 The detailed proof of Proposition \ref{ProUserSINRCoverageProbability}  can  be found in Appendix \ref{Appen3:ProUserSINRCoverageProbability}.
\end{proof}

%

\section{APT and ASE of   Cache-enable mABHetNets }
APT and ASE are applied as two significant metrics to measure the network performance. The APT focuses on the average user QoS requirement in terms of data rate while ASE mainly  is used to measure the average network spectral efficiency. In this section, we first derive the APT. Then, we investigate the network ASE  in bps/Hz/m$^2$ and     analyze it.

\subsection{ APT of   Cache-enable mABHetNets}
APT captures the average number of bits that can be received by the user per unit area   per unit bandwidth given a pre-designated threshold $\gamma$ \cite{Throuhgput1}. The definition of APT  is
\begin{equation*}
   \mathcal{R}\left(  \gamma\right) =\lambda_k W_k\log _{2}\left(1+\gamma\right) \mathbb{P}\left\{\operatorname{SINR} \geq \gamma\right\}
\end{equation*}
where $\lambda_k(k=s,m)$ is the density of SBS or MBS and $W_k$ is allocated bandwidth to the user. $\gamma$ is the user's SINR requirement.

In cache-enable mABHetNets, APT is determined by cache capacity, bandwidth partition and SINR threshold. Then we
let APT  of cache-enable mABHetNets   denoted by
\begin{equation}\label{NetworkThroughput}
  \mathcal{R}(\eta,C,\gamma_0) = \mathcal{R}_s(\eta,C,\gamma_0)+\mathcal{R}_m(\eta,\gamma_0)
\end{equation}
where
$\mathcal{R}_s$ and $\mathcal{R}_m$ are   APT of  the SBS tier and MBS tier. $C$ is the cache capacity of SBS. $\eta$ is the bandwidth partition between the access link and the backahul link.  $\gamma_0$ is the SINR threshold to guarantee the user throughout requirement. The following corollaries will give the detailed APT expression.

\subsubsection{ APT of SBS tier}
For a user associated with a SBS, the transmission path include the  access link between SBS and user and the backhaul link bewteen the SBS and MBS. Besides, in cache-enabled mABHetNets, the caches in SBS tier also influence the file delivery in the transmission path. When the files are cached at the SBS, then the files can be delivered to user directly. At this time, the wireless backhaul between the SBS and the MBS will not be used. Otherwise, the uncache files will be delivered through the wireless backhaul link. Given the caches in SBS,   we first give APT of SBS tier.
\begin{corollary}\label{corolAPTSBS}
Since the transmission  can be LoS or NLoS in wireless access link and wireless backhaul link for user associated with SBS tier, there are four cases in the SBS-tier throughput. Then
\begin{equation}\label{APTSBS}
   \mathcal{R}_s(\eta,C,\gamma_0) = \mathcal{R}_{s}^{ll}+\mathcal{R}_{s}^{ln}+\mathcal{R}_{s}^{nl}+\mathcal{R}_{s}^{nn}
\end{equation}
where $\mathcal{R}_{s}^{ll},\mathcal{R}_{s}^{ln},\mathcal{R}_{s}^{nl}$ and $\mathcal{R}_{s}^{nn}$ denotes the network throughput when the wireless SBS link and the wireless backhaul link are both LoS,  the wireless SBS link is LoS and the wireless backhaul link is NLoS, the wireless SBS link is NLoS and the wireless backhaul link is LoS and  the wireless SBS link is NLoS and the wireless backhaul link is NLoS, respectively.Then $\mathcal{R}_{s}^{ll}(\eta,C,\gamma_0)  = \min \{ \lambda_s\eta W\log_2(1+\gamma_0)P^{cov}_{s,\mathrm{L}}(\gamma_0),\frac{1 }{ 1-p_h }\lambda_m W(1-\eta)\log_2(1+\gamma_0)P^{cov}_{bh,\mathrm{L}}(\gamma_0)\}$. Symbol $\min\{\}$ means that the minimum value between the wireless access link rate and the wireless backhaul link rate. Following the same logic,
$
    \mathcal{R}_{s}^{ln}    = \min \{ \lambda_s\eta W\log_2(1+\gamma_0)P^{cov}_{s,\mathrm{L}}(\gamma_0),\frac{1 }{ 1-p_h }\lambda_m W(1-\eta)\log_2(1+\gamma_0)P^{cov}_{bh,\mathrm{L}}(\gamma_0)\},
    \mathcal{R}_{s}^{nl}  = \min \{ \lambda_s\eta W\log_2(1+\gamma_0)P^{cov}_{s,\mathrm{L}}(\gamma_0),$ ~~~$\frac{1 }{ 1-p_h }\lambda_m W(1-\eta)\log_2(1+\gamma_0)P^{cov}_{bh,\mathrm{L}}(\gamma_0)\},
    \mathcal{R}_{s}^{nn}    = \min \{ \lambda_s\eta W\log_2(1+\gamma_0)P^{cov}_{s,\mathrm{L}}(\gamma_0),\frac{1 }{ 1-p_h }\lambda_m W(1 -\eta)\log_2(1+\gamma_0)P^{cov}_{bh,\mathrm{L}}(\gamma_0)\}
$. Note that, $p_h = p_h(C) =\frac{\sum_{f=1}^{C}f^{-\gamma_p}}{\sum_{g=1}^{F}g^{-\gamma_p}}$ is the cache hit ratio in the SBS tier. $(1-p_h)$ reflects the probability   that the files that are not cached in SBS tier will delivered through the backhaul link.  $P^{cov}_{s,\mathrm{L}}(\gamma_0)$ and $P^{cov}_{s,\mathrm{NL}}(\gamma_0)$ are the SINR coverage probability that the user is associated with the SBS via LoS and NLoS path in the Proposition \ref{ProUserSINRCoverageProbability}.
\end{corollary}

\subsubsection{ APT of MBS tier}
Since the signal is directly transmit by MBS to user via access link, similar to  Corollary \ref{corolAPTSBS}, we can easily give the expression of APT of MBS tier.
\begin{corollary}\label{corolAPTMBS}
It is easy to obtain the average throughput of the MBS tier:
    \begin{align}\label{Throughput}
    \mathcal{R}_m(\eta,\gamma_0)
    &   =\lambda_m\eta W\log_2(1+\gamma_0)P^{cov}_{m,\mathrm{L}}(\gamma_0)    + \lambda_m\eta W\log_2(1+\gamma_0)P^{cov}_{m,\mathrm{NL}}(\gamma_0)
    \end{align}
where $P^{cov}_{m,\mathrm{L}}(\gamma_0)$ and $P^{cov}_{m,\mathrm{NL}}(\gamma_0)$ are the SINR coverage probability that the user is associated with the MBS via LoS and NLoS path in the Proposition \ref{ProUserSINRCoverageProbability}.
\end{corollary}

\subsection{User Spectral Efficiency   Distribution  of Cache-enable mABHetNets}
Before giving the ASE of mABHetNets, the user spectral efficiency should be first analyzed.
For a user associated with a SBS, the transmission rate of the user is not only related with the access  rate from SBS, but also with the backhaul rate from MBS. When the files are cached at the SBS, then the files can be delivered to user directly. At this time, the wireless backhaul between the SBS and the MBS will not be used. Otherwise, the uncache files will be delivered first through the wireless backhaul link then the access link.  Let $\rho$ be the required user spectral efficiency and the $p_h$ in  (\ref{HitRatio}) be the file hit ratio at the SBS. Then the required wireless backhaul link  spectral efficiency is $(1-p_h)\rho$. The definition of the user spectral efficiency distribution is  $\mathbb{P}\left[R>\rho \right]$. Then in the following the proposition, we give the detailed user spectral efficiency distribution of the user.
\begin{proposition}\label{ProRateCoverageDistribution}
1)The   spectral efficiency  distribution of a user associated with an SBS:
When a user is associated with a SBS and the files are cached in SBS, the cached file will be delivered  directly   from SBS by the wireless access link. Otherwise, the uncached file will first delivered  by  the MBS through the backhaul link then through the access link. Since the wireless access link and wireless backhaul link can be LoS or NLOS, then four cases will exist as follows:
    \begin{align}
        &\text{LoS access link \& LoS backhaul link:}\mathbb{P}\left[R_{s,1}>\rho \right]=\mathbb{P}_{s,\mathrm{L}}^{cov}(2^{\frac{\rho}{\eta}}-1|r_{s} )\cdot \mathbb{P}_{bh,\mathrm{L}}^{cov}(2^{\frac{(1-p_h)\rho}{(1-\eta) }}-1|r_{bh} ),\nonumber\\
        &\text{LoS access link \& NLoS backhaul link:}\mathbb{P}\left[R_{s,2}>\rho \right]=\mathbb{P}_{s,\mathrm{L}}^{cov}(2^{\frac{\rho}{\eta}}-1|r_{s} ) \cdot   \mathbb{P}_{bh,\mathrm{NL}}^{cov}(2^{\frac{(1-p_h)\rho}{(1-\eta) }}-1|r_{bh} ),\nonumber\\
        &\text{NLoS access link \& LoS backhaul link:}\mathbb{P}\left[R_{s,3}>\rho \right]=\mathbb{P}_{s,\mathrm{NL}}^{cov}(2^{\frac{\rho}{\eta}}-1|r_{s} ) \cdot   \mathbb{P}_{bh,\mathrm{L}}^{cov}(2^{\frac{(1-p_h)\rho}{(1-\eta) }}-1|r_{bh} ),\nonumber\\
        &\text{NLoS access link \&NLoS backhaul link:}\mathbb{P}\left[R_{s,4}>\rho \right]=\mathbb{P}_{s,\mathrm{NL}}^{cov}(2^{\frac{\rho}{\eta}}-1|r_{s} )     \mathbb{P}_{bh,\mathrm{NL}}^{cov}(2^{\frac{(1-p_h)\rho}{(1-\eta) }}-1|r_{bh} ).\nonumber
    \end{align}
 where $R_{s,1}$ denotes the spectral efficiency of the link when the wireless access link is LoS and the backahul link is NLoS. Similarly,   $R_{s,2}$,  $R_{s,2}$ and $R_{s,4}$ are the other cases. The file hit ratio is $p_h=p_{h}(C)=\frac{\sum_{f=1}^{C}f^{-\gamma_p}}{\sum_{g=1}^{F}g^{-\gamma_p}}$. The transmission power of a SBS is $P_s^{tr}=P_s^{tr}(C)=P'_s-w'_{\mathrm{ca}}C$, which is related with the SBS cache capacity $C$. Besides, $\mathbb{P}_{s,\mathrm{L}}^{cov}{(\cdot)}$ and $\mathbb{P}_{s,\mathrm{NL}}^{cov}{(\cdot)}$ are the SINR distributions of the SBS-user via LoS and NLoS path. $\mathbb{P}_{bh,\mathrm{L}}^{cov}{(\cdot)}$ and $\mathbb{P}_{bh,\mathrm{NL}}^{cov}{(\cdot)}$ are the SINR distributions of the backhaul link between the SBS and MBS via LoS and NLoS path in Proposition \ref{ProUserSINRCoverageProbability}.

2)The spectral efficiency   distribution   of a user associated with an MBS:
    \begin{align}
      &\text{for LoS based the access link  :~~~~~}
      \mathbb{P}\left[R_{m,1}>\rho \right]=\mathbb{P}_{m,\mathrm{L}}^{cov}(2^{\frac{\rho}{\eta}}-1|r_{m}), \\
      &\text{for NLoS based the access link  :~~~}
      \mathbb{P}\left[R_{m,2}>\rho\right]=\mathbb{P}_{m,\mathrm{NL}}^{cov}(2^{\frac{\rho}{\eta}}-1|r_{m}),
    \end{align}
where    $\mathbb{P}_{m,\mathrm{L}}^{cov}{(\cdot)}$ and $\mathbb{P}_{m,\mathrm{NL}}^{cov}{(\cdot)}$ are the SINR distributions of the MBS-user via LoS and NLoS path in Proposition \ref{ProUserSINRCoverageProbability}.
\end{proposition}

\begin{proof}
  The proof procedure can be found in the Appendix \ref{Appen4:ProRateCoverageDistribution}.
\end{proof}

\subsection{ASE   of Cache-enabled mABHetNets}
According to the definition of the ASE in \cite{Throuhgput1}, ASE can be expressed as follows:
\begin{align}\label{AreaSpectralEfficiency}
 \mathcal{A}(\eta,C)
 &=\mathcal{A}_s(\eta,C)+\mathcal{A}_m(\eta)
 =\lambda \mathbb{E}\left[R_s \right]+\lambda \mathbb{E}\left[R_m \right]\\
 &=\lambda \mathbb{E}_r\left[\int_{0}^{\infty}\mathbb{P}\left[R_s>\rho|r\right]\mathrm{d}\rho\right]
 +\lambda \mathbb{E}_r\left[\int_{0}^{\infty}\mathbb{P}\left[R_m>\rho|r\right]\mathrm{d}\rho\right]\nonumber
\end{align}
where $\eta$ is the part of the spectrum allocated to the access link and $C$ is the cache capacity of the SBS. $\mathcal{A}_s(\eta,C)$ is the ASE of SBS tier and $\mathbb{P}\left[R_s>\rho\right]$ is  user's spectral efficiency distribution when the user is associated with the SBS tier. $\mathcal{A}_m(\eta)$ is the ASE of MBS tier and $\mathbb{P}\left[R_m>\rho\right]$ is the rate   distribution when the user is associated with the MBS tier.

Next, we will derive the $\mathcal{A}_s(\eta,C)$ and $\mathcal{A}_m(\eta )$, respectively.
\begin{corollary} \label{ProBSAreaSpectralEfficiency}
Based on the four cases of the user's spectral efficiency   distribution in Proposition \ref{ProRateCoverageDistribution}, the
ASE  of SBS tier is
    \begin{align}\label{SBSAreaSpectralEfficiency}
        &\mathcal{A}_s(\eta,C)=\lambda (\int_{0}^{\infty}\int_{0}^{\infty}\left[\int_{0}^{\infty}\mathbb{P}\left[R_{s,1}>\rho\right]\mathrm{d}
        \rho\right]F_{s}^\mathrm{L}(r_{s})F_{bh}^\mathrm{L}{(r_{bh})}\\
        &+\left[\int_{0}^{\infty}\mathbb{P}\left[R_{s,2}>\rho\right]\mathrm{d}
        \rho\right]F_{s}^\mathrm{L}(r_{s})F_{bh}^\mathrm{NL}{(r_{bh})} +\left[\int_{0}^{\infty}\mathbb{P}\left[R_{s,3}>\rho\right]\mathrm{d}
        \rho\right]F_{s}^\mathrm{NL}(r_{s})F_{bh}^\mathrm{L}{(r_{bh})}\nonumber\\
        &+\left[\int_{0}^{\infty}\mathbb{P}\left[R_{s,4}>\rho\right]\mathrm{d}
        \rho\right]F_{s}^\mathrm{NL}(r_{s})F_{bh}^\mathrm{NL}{(r_{bh})}\nonumber
        \mathrm{d}r_{s}\mathrm{d}r_{bh})
         \end{align}
where cache capacity $C$ exists in the spectral efficiency $\mathbb{P}\left[R_{s,1}>\rho\right]$, $\mathbb{P}\left[R_{s,2}>\rho\right]$, $\mathbb{P}\left[R_{s,3}>\rho\right]$ and $\mathbb{P}\left[R_{s,4}>\rho\right]$ (proposition \ref{ProRateCoverageDistribution}).

And  ASE  of MBS tier is
        \begin{align}\label{MBSAreaSpectralEfficiency}
            &\mathcal{A}_m(\eta )= \lambda(\int_{0}^{\infty}
            \left[\int_{0}^{\infty}\mathbb{P}\left[R_{m,1}>\rho\right]\mathrm{d}
            \rho\right]F_{m}^\mathrm{L}(r_{m}) +\left[\int_{0}^{\infty}\mathbb{P}\left[R_{m,2}>\rho\right]\mathrm{d}
            \rho\right]F_{m}^\mathrm{NL}(r_{m})  \mathrm{d}r_m)
         \end{align}
where $\mathbb{P}\left[R_{m,1}>\rho \right]=\mathbb{P}_{m,L}^{cov}(2^{\frac{\rho}{\eta}}-1|r_{m} )$ and $ \mathbb{P}\left[R_{m,2}>\rho \right]=\mathbb{P}_{m,\mathrm{NL}}^{cov}(2^{\frac{\rho}{\eta}}-1|r_{m} )$.
\end{corollary}

\begin{remark}
  From (\ref{AreaSpectralEfficiency}) and (\ref{SBSAreaSpectralEfficiency}), when the file hit ratio is improved,  more files can be delivered from the SBS tier to users directly and  the network throughput can be increased.  Therefore, ASE can be increased. On one hand, the most efficient case is the file popularity can be more centralized, which means less files own more high popularity(i.e.,$p_h=p_{h}(C)=\frac{\sum_{f=1}^{C}f^{-\gamma_p}}{\sum_{g=1}^{F}g^{-\gamma_p}}$ has a higher $\gamma_p$ ). On the other hand, the cache capacity is increased to raise the hit ratio $p_h$. However, we can  see that more power is consumed for cache and the transmission power is decreased, which decrease the ASE. All the analysis will be verified in the numerical results.
\end{remark}
\subsection{A Special  Case:  the Noise-Limited mABHetNets}

 Based on the research in  \cite{NoiseLimit,NoiseLimit1,NoiseLimit2,LowDense}, under the lower SBS density, mABHetNets will be    noise-limited. In mmWave communication systems, mmWave has a high signal transmission strength in LoS   links. However, when the density of  SBS is lower, the number of  LoS link will decrease.  The LoS based interferences  of the intra-tier and cross-tier can be omitted. Therefore, the noise has a more important influence on the mmWave signal transmission. According to such situation, we want to investigate the ASE in cache-enabled mABHetNets in this noise-limited case.\cite{LowDense}.

 Here, the density $\lambda_m$ of MBS in a real mABHetNets is low and fixed. Therefore, in this section, the low density is referring in particular to the   SBS denstiy$\lambda_s$.

\begin{corollary}
\label{NoiseASE}
The ASE of mABHetNets in the noise-limited case is
\begin{equation}\label{NoiseASEExpression}
\mathcal{A}^{\mathrm{Noi}}(\eta,C)=\mathcal{A}_{s}^{\mathrm{Noi}}(\eta,C)+\mathcal{A}_{m}^{\mathrm{Noi}}(\eta)
\end{equation}
where $\mathcal{A}_{s}^{\mathrm{Noi}}(\eta,C)$ and $\mathcal{A}_{m}^{\mathrm{Noi}}(\eta)$ are the ASEs of SBS tier and MBS tier in the noise-limited case, respectively.

With  $I_r=I_m=I'_r=I'_m=0$ in  the general ASE expression of SBS tier (\ref{SBSAreaSpectralEfficiency}), ASE  of SBS tier in the noise-limited mABHetNets is
\begin{align}\label{ASESBSNoiseLimited}
  \mathcal{A}_s^{\mathrm{Noi}}(\eta,C)
  &= \lambda  (\int_{0}^{\infty}\int_{0}^{\infty} \int_{0}^{\infty}A_1(\eta,C) F_s^{\mathrm{L}}(r_s) F_{bh}^{\mathrm{L}}(r_{bh})  +   A_2(\eta,C) F_s^{\mathrm{L}}(r_s)F_{bh}^{\mathrm{NL}}(r_{bh}) \\
  & A_3(\eta,C) F_s^{\mathrm{NL}}(r_s)F_{bh}^{\mathrm{L}}(r_{bh})
  +
   A_4(\eta,C) F_s^{\mathrm{NL}}(r_s)F_{bh}^{\mathrm{NL}}(r_{bh})\mathrm{d}\rho\mathrm{d}r_s\mathrm{d}r_{bh})\nonumber
\end{align}
where
    $\scriptstyle  A_1(\eta,C) =\exp \left(\frac{-(2^\frac{\rho}{\eta }-1) N_{0}}{P_s^{tr}  B_s A_{\mathrm{L}} r_s^{-\alpha_{\mathrm{L}}} }+\frac{-(2^\frac{(1-p_h(C))\rho}{(1-\eta)}-1)N_{0}}{P_{m}^{tr} B_m A_{\mathrm{L}} r_{bh}^{-\alpha_{\mathrm{L}}} }\right)$,$\scriptstyle A_2(\eta,C)=\exp \left(\frac{-(2^\frac{\rho}{\eta }-1) N_{0}}{P_s^{tr}B_s A_{\mathrm{L}} r_s^{-\alpha_{\mathrm{L}}} }+\frac{-(2^\frac{(1-p_h(C))\rho}{(1-\eta)}-1)N_{0}}{P_{m}^{tr} B_m A_{\mathrm{NL}} r_{bh}^{-\alpha_{\mathrm{NL}}} }\right),~~~~~~~~$
    $ ~~~~~~~~\scriptstyle A_3(\eta,C)=\exp \left(\frac{-(2^\frac{\rho}{\eta }-1) N_{0}}{P_s^{tr}B_s A_{\mathrm{NL}} r_s^{-\alpha_{\mathrm{NL}}} }+\frac{-(2^\frac{(1-p_h(C))\rho}{(1-\eta)}-1)N_{0}}{P_{m}^{tr} B_m A_{\mathrm{L}} r_{bh}^{-\alpha_{\mathrm{L}}} }\right) $,
  $\scriptstyle A_4(\eta,C)=\exp \left(\frac{-(2^\frac{\rho}{\eta }-1) N_{0}}{P_s^{tr} B_sA_{\mathrm{NL}} r_s^{-\alpha_{\mathrm{NL}}} }+\frac{-(2^\frac{(1-p_h(C))\rho}{(1-\eta)}-1)N_{0}}{P_{m}^{tr}  B_mA_{\mathrm{NL}} r_{bh}^{-\alpha_{\mathrm{NL}}} }\right)$,
  $p_h=p_{h}(C)=\frac{\sum_{f=1}^{C}f^{-\gamma_p}}{\sum_{g=1}^{F}g^{-\gamma_p}}$ $P_s^{tr}=P'_s-w'_{\mathrm{ca}}C$.

  In the same way, from the general ASE expression of MBS tier in (\ref{MBSAreaSpectralEfficiency}), the  ASE  of MBS tier in noise-limited environment is
  \begin{small}
  \begin{align}\label{ASEMBSNoiseLimited}
   & \mathcal{A}_m^{\mathrm{Noi}}(\eta)=\lambda(\int_{0}^{\infty} \int_{0}^{\infty}\exp \left(\frac{-(2^{\frac{\rho}{\eta}}-1)N_{0}}{P_{m}^{tr}B_mA_{\mathrm{L}} r_m^{-\alpha_{\mathrm{L}}} }\right) F_m^{\mathrm{L}}(r_m)   +
   \exp \left(\frac{-(2^{\frac{\rho}{\eta}}-1) N_{0}}{P_{m}^{tr}B_m A_{\mathrm{NL}} r_m^{-\alpha_{\mathrm{NL}}} }\right)  F_m^{\mathrm{NL}}(r_m)\mathrm{d}\rho\mathrm{d}r_m)
  \end{align}
  \end{small}
  Note that, since the density of MBS tier is low, the transmission of the MBS tier can be considered as noise-limited.
\end{corollary}

\begin{remark}
  Based on the SINR expression in (\ref{SBSUSERSINR}) , (\ref{MBSUSERSINR}) and (\ref{M-SBSSINR}), the SINR without interference  in the noise-limited mABHetNets is larger than the original one, the  approximated ASE in (\ref{NoiseASEExpression}) is actually a tighter upper bound of the original ASE in (\ref{AreaSpectralEfficiency}).We show later in the numerical results that ignoring the interference in   the noise-limited mABHetNets    introduces a negligible error.
\end{remark}

\subsection{A   Special Case: Interference-Limited}
When the density of SBS is very high, the LoS link based   desired signal and interference will become dominant\cite{InterferenceLimted,InterferenceLimted1}. At this time the NLoS based signal and interference and   the noise received at the typical user are usually omitted. Such cache-enabled mABHetNets is called  interference-limited one.  Then, in the interference-limited mABHetNets,   with  $N_0=0$ and ignored NLoS transmission in (\ref{SBSUSERSINR}) and (\ref{MBSUSERSINR}),
ASE  of SBS tier is given in below lemma. Since the interference-limited case exist in the denser SBS and not related to the MBS tier of lower density.
\begin{proposition}\label{InterferenceASE}
The ASE in the interference-limited mABHetNets is
\begin{equation}\label{InterferenceASEExpression}
\mathcal{A}^{\mathrm{Int}}(\eta,C)=\mathcal{A}_{s}^{\mathrm{Int}}(\eta,C)+\mathcal{A}_{m}^{\mathrm{Int}}(\eta)
\end{equation}
where $\mathcal{A}_{s}^{\mathrm{Int}}(\eta,C)$ and $\mathcal{A}_{m}^{\mathrm{Int}}(\eta)$ are the ASEs of SBS tier and MBS tier in the interference-limited case, respectively.

Based on the general SBS ASE in (\ref{SBSAreaSpectralEfficiency}), the SBS ASE in the interference-limited case is
\begin{align}\label{ASEinterferenceLimited}
  &\mathcal{A}_{s}^{\mathrm{Int}}(\eta,C)\\
  &=\lambda
  (\int_{0}^{\infty} \int_{0}^{\infty} \int_{0}^{\infty}  \mathcal{\overline{L}}_{I_{s,m}}^{\mathrm{L,int}}\left((2^{\frac{\rho}{\eta}}-1)r_s^{\alpha_{\mathrm{L}}}\right)  \mathbb{P}_{bh,\mathrm{L}}^{cov}(2^{\frac{(1-p_h)\rho}{(1-\eta) }}-1|r_{bh}) F_{s}^{\mathrm{L}}(r_s)F_{bh}^\mathrm{L}{(r_{bh})}\mathrm{d}\rho\mathrm{d}r_s\mathrm{d}r_{bh}\nonumber\\
  &+
  \int_{0}^{\infty} \int_{0}^{\infty} \int_{0}^{\infty}  \mathcal{\overline{L}}_{I_{s,m}}^{\mathrm{L,int}}\left((2^{\frac{\rho}{\eta}}-1)r_s^{\alpha_{\mathrm{L}}}\right)  \mathbb{P}_{bh,\mathrm{NL}}^{cov}(2^{\frac{(1-p_h)\rho}{(1-\eta) }}-1|r_{bh}) F_{s}^{\mathrm{L}}(r_s)F_{bh}^\mathrm{NL}{(r_{bh})}\mathrm{d}\rho\mathrm{d}r_s\mathrm{d}r_{bh}\nonumber
\end{align}
where $\mathcal{\overline{L}}_{I_{s,m}}^{\mathrm{L,int}}(\cdot)$  are the  Laplace transform of the cumulative   interference from the SBS tier in the interference-limited case. $\mathbb{P}_{bh,\mathrm{NL}}^{cov}(\cdot)$ and $\mathbb{P}_{bh,\mathrm{NL}}^{cov}(\cdot)$ are the SINR distributions in (\ref{BHSINRCoverageProbabilityLoS}) and (\ref{BHSINRCoverageProbabilityNLoS}), respectively.

It is noted that, the ASE of MBS in the interference-limited case is the same as  the general expression in (\ref{MBSAreaSpectralEfficiency}) since the MBS density is not changed. Namely,
\begin{align}
 \mathcal{A}_m^{\mathrm{Int}}(\eta)=\mathcal{A}_m(\eta).
 \end{align}
\end{proposition}
\begin{proof}
  The detailed proof can be found in Appendix \ref{Appen5:Inteference}.
\end{proof}

\begin{remark}
  Based on the SINR expression in (\ref{SBSUSERSINR}) , (\ref{MBSUSERSINR}) and (\ref{M-SBSSINR}), SINR with no NLoS interference is larger than the original one,    the  approximated ASE in (\ref{InterferenceASEExpression})another tighter upper bound of the original ASE in (\ref{AreaSpectralEfficiency}). We show later in the numerical
results that ignoring the NLoS interference and noise in  mABHetNets introduces a negligible
error.
\end{remark}




\section{numerical  results}
In this section, we use numerical results to validate and evaluate of APT and   ASE of the cache-enabled mHetNets. We further   study   APT and ASE under different network scenarios and cache parameters.

\subsection{Parameter Setting}
%
The density of the SBS $\lambda_s$ and the MBS $\lambda_m$ are $10^{-4}$BSs/m$^2$  and $10^{-5}$BSs/m$^2$. Note that  the density of the users is assumed to be sufficiently larger than that of the BS so that each BS has at least one associated user in its coverage. The density of the users $\lambda_s=3\times10^{-4}$users/m$^2$. The Zipf distribution parameter $\gamma_p$ of file popularity  is 0.6\cite{Zipf}. Based on   \cite{BSPowerModel1}, we can assume that each file unit has the same size of 4MB. The number of files in the file library is 1000 file units. The cache capacity of SBS    is 100 file units. To reflect the caching power model, we adopt the caching power coefficient $\omega_{ca}$ which is $2.5\times10^{-9}$W/bit\cite{BSPowerModel}. According to the simulation requirement, the total power of SBS and MBS is set as 9.1W and 610W to maintain the   transmission power consumption and caching power consumption. Other default simulation configurations are listed in Table \ref{parameters}, based on 3GPP specification and literatures\cite{Partition3,3GPP:Spec,PathLoss,FixPower,SBSPowerScaling}. All the above settings will be changed according to different scenarios.
\renewcommand\arraystretch{0.7}
\begin{table}[h]
\small
\centering
\caption{Simulation parameters}
\label{parameters}
 \begin{tabular}{cp{2cm}c}
 \hline
 \hline
 Parameters                                                   &&  Values\\
 \hline
 Total mmWave spectrum bandwidth $W$                              &&400 MHz\\
 LoS pathloss parameters $A_\mathrm{L},\alpha_\mathrm{L}$                     &&$10^{-10.38},~2.09$\\
 NLoS pathloss parameters $A_{\mathrm{NL}},\alpha_{\mathrm{NL}}$              &&$10^{-14.54},~3.75$\\
 Noise Power    $N_0$                                               &&5 dB\\
 Fixed circuit power at MBS                                         &&10.16W\\
 Fixed circuit power at SBS                                         &&0.1W\\
 Power amplifier and cooling coefficient   $\rho_s$and $\rho_m$                 &&4,~15.13\\
 Association biases of SBS and MBS $B_s$ and $B_m$                              &&10,~1\\
 Blockage rate $\beta$                                              &&2.7$\times$10$^{-2}$\\
 \hline
 \hline
 \end{tabular}
\end{table}

\subsection{APT of mABHetNets}
\begin{figure}[ht]
\setlength{\abovecaptionskip}{0.cm}

\setlength{\belowcaptionskip}{-0.cm}
  \centering
  \subfigure[]{\includegraphics[width=2.1 in]{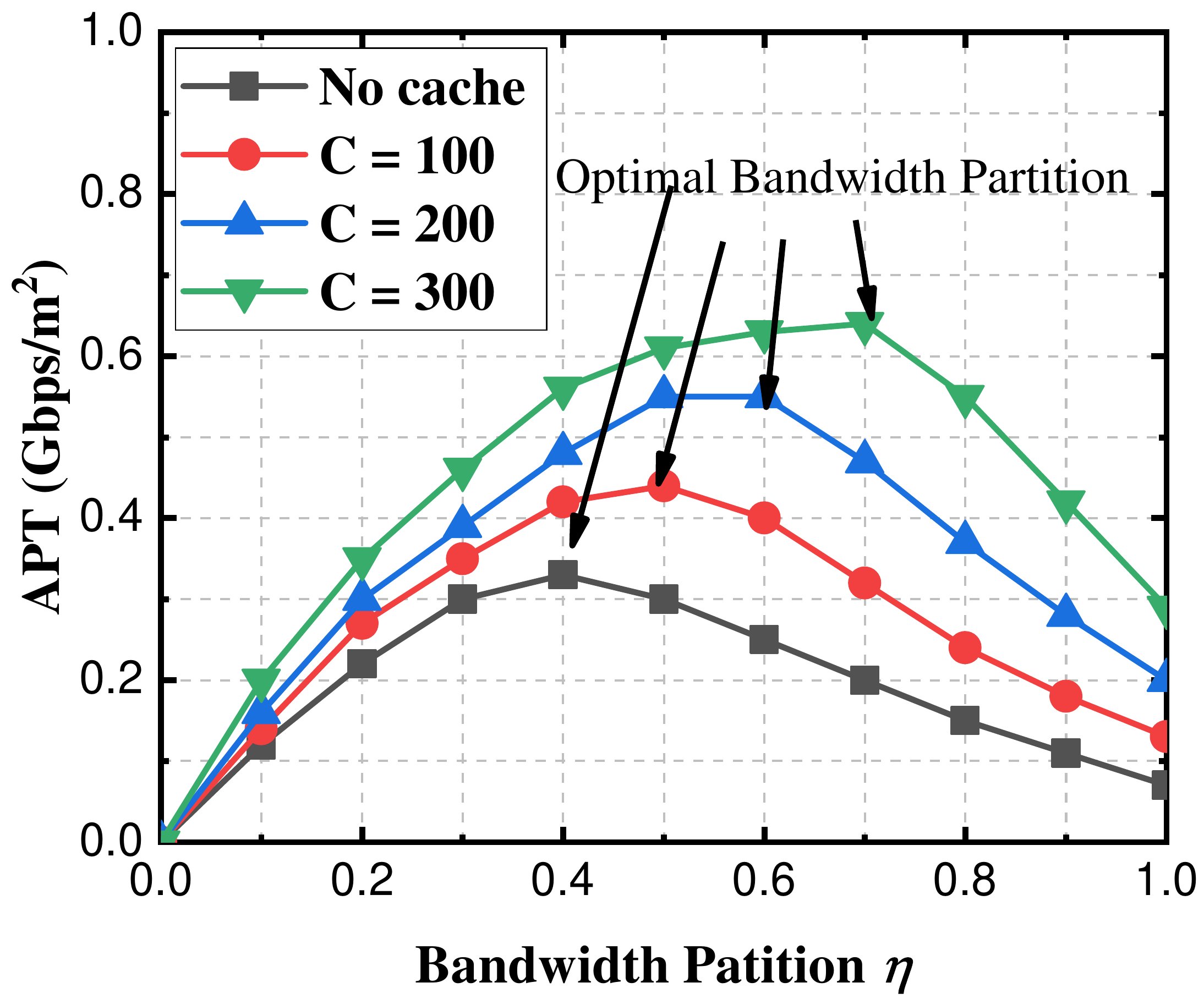}\label{APTEta}}
  \subfigure[]{\includegraphics[width=2.1 in]{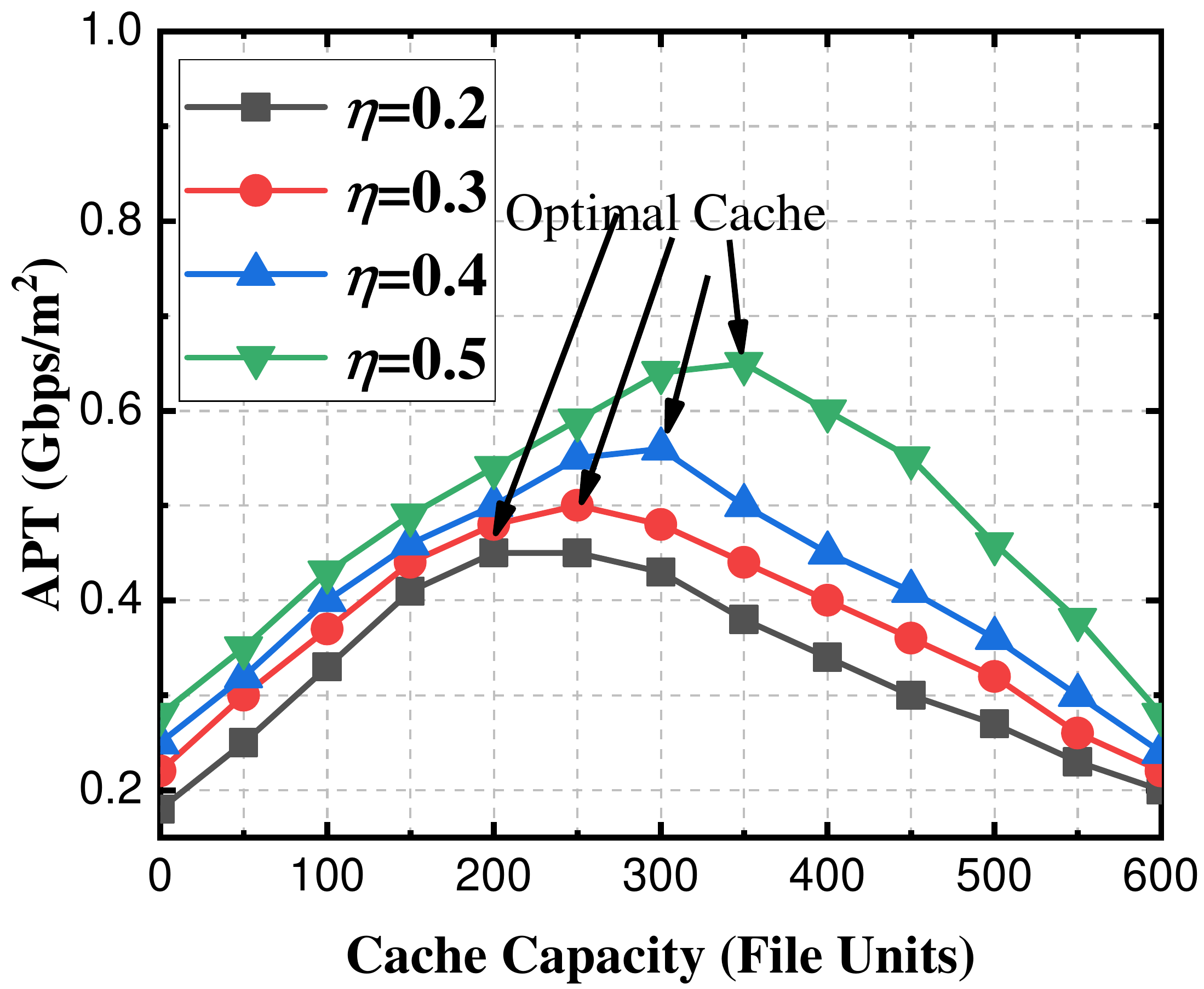}\label{APTCache}}
  \subfigure[]{\includegraphics[width=2.05 in]{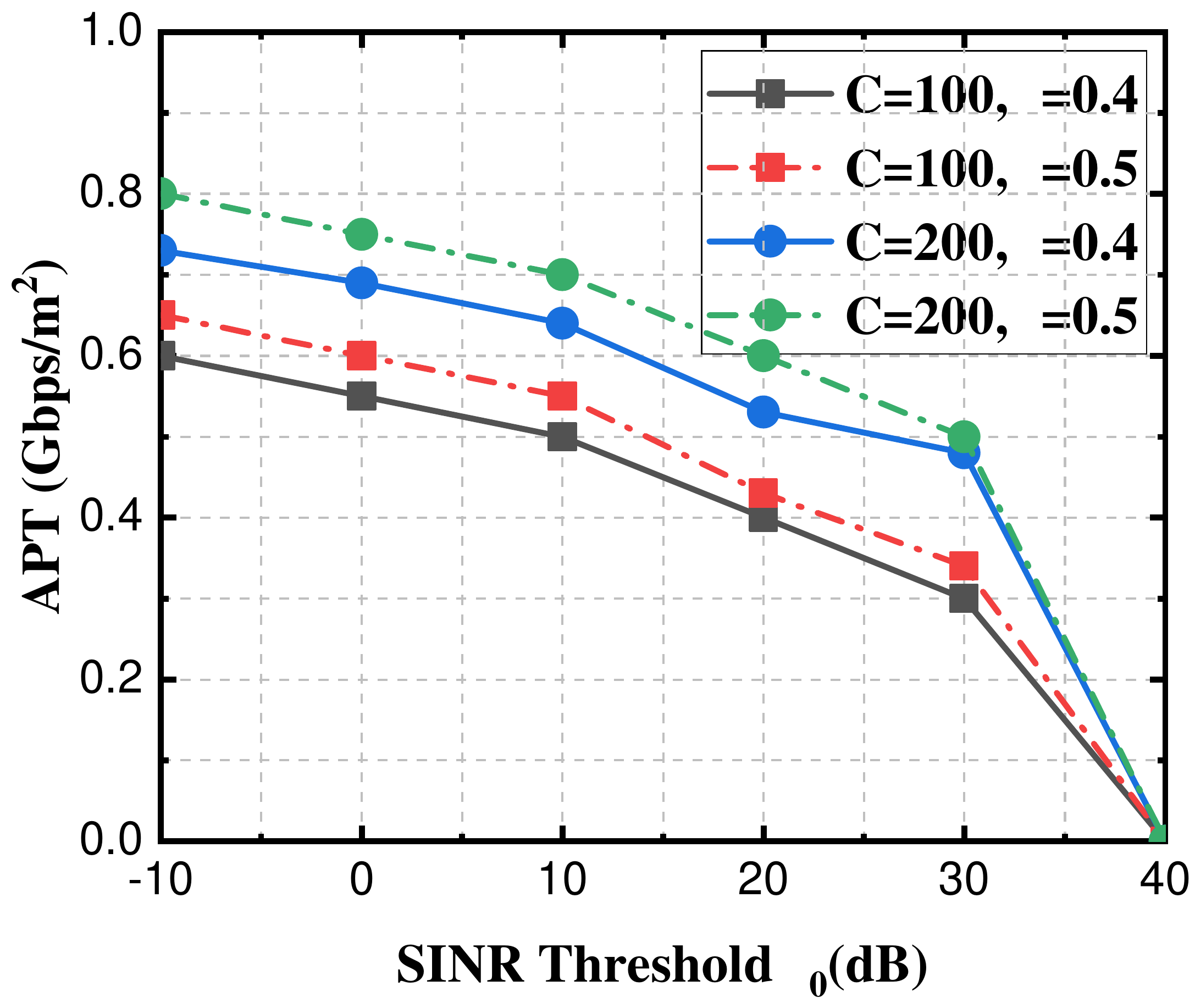}\label{APTGamma}}
  \captionsetup{font={small}}
  \caption{APT of the mABHetNets under different (a)bandwidth partitions( $\gamma_0$ is 10dB), (b)cache capacities($\gamma_0$ is 10dB) and (c)SINR thresholds.}
  \label{APT}
\end{figure}
To verity the APT performance in cache-enabled mHeNets, APTs under different network parameter or caching parameters are given. In Fig. \ref{APTEta},   APT will increase as the bandwidth partition increases and then it decreases. That means there exist the optimal bandwidth in APT. This is because, when the backhaul spectrum bandwidth is enough, transferring some bandwidth to the  access can increase the throughput of the user. Howerver, when more bandwidth is used in access, the backhaul link throughput cannot maintain the backhaul of the access throughput and the total throughput is reduced. In Fig. \ref{APTCache}, APT also increase with the increasing cache capacity, when the backhaul throughput  is limited with lower backhual bandwidth. As more cache files can improver the cache hit ratio of SBS and   less files will  use backhaul resource. Then more files can be obtained by access without  bandwidth  and the ASE is improved. However, when the cache capacity is over 600, APT is zero. Such result is because that the maximum  power of SBS is limited and more cache capacity consumes more power and the transmission power is reduced. The reduced transmission power will decrease the data rate and APT. Since the APT is related to the user rate requirement, APT is shown under different SINR thresholds in Fig. \ref{APTGamma}. High SINR threshold will decrease APT of user. This is due to the fact that  the path loss and fading make the received power lower and the received SINR lower than the threshold. However, at the same SINR threshold, more cache  capacity and more bandwidth partition can cause more APT.

\subsection{ASE Performance under bandwidth partition and Cache Capacity}
\begin{figure}[ht]
  \centering
  \subfigure[]{\includegraphics[width=2.5 in]{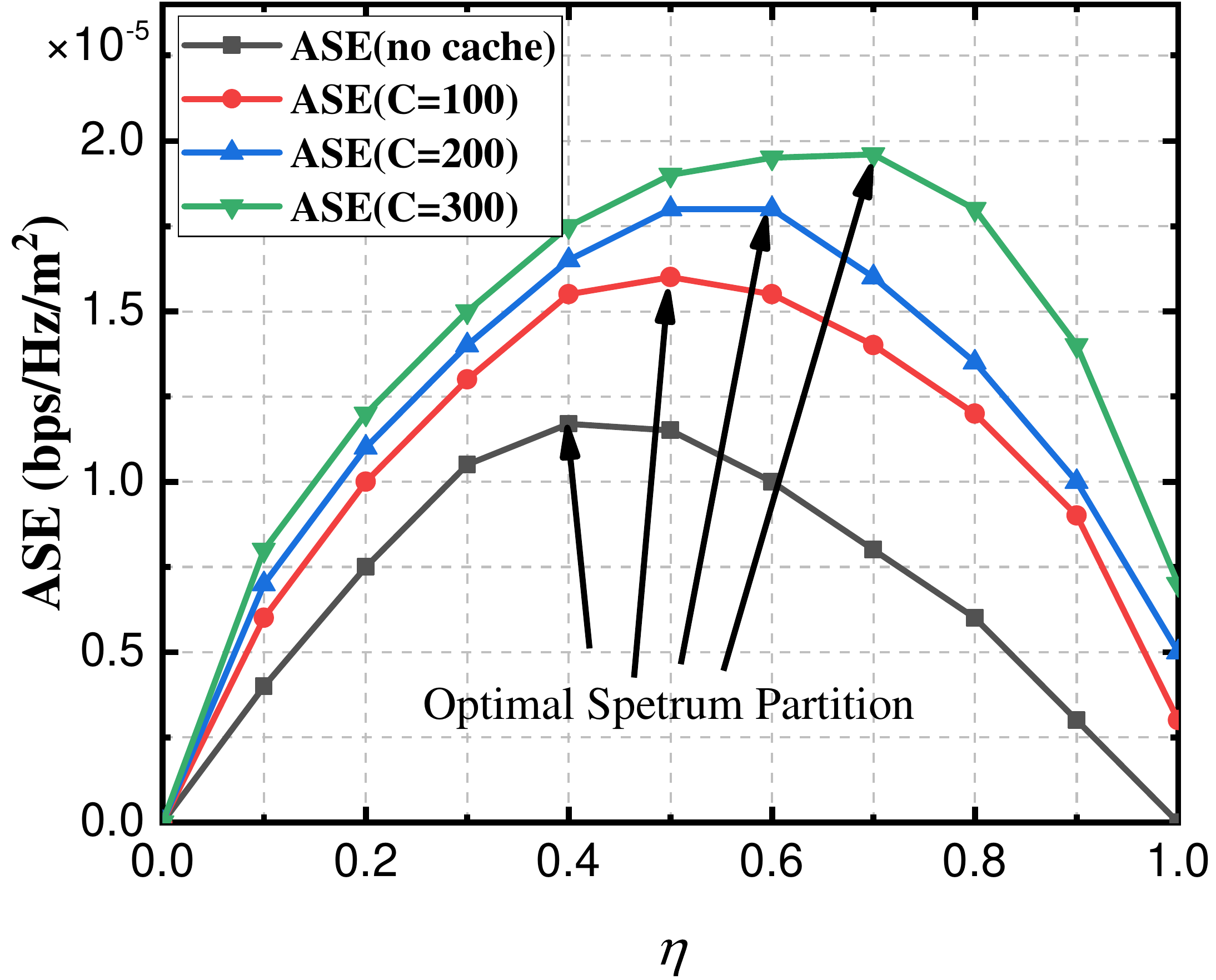}\label{ASEPartition}}
  \subfigure[]{\includegraphics[width=2.5 in]{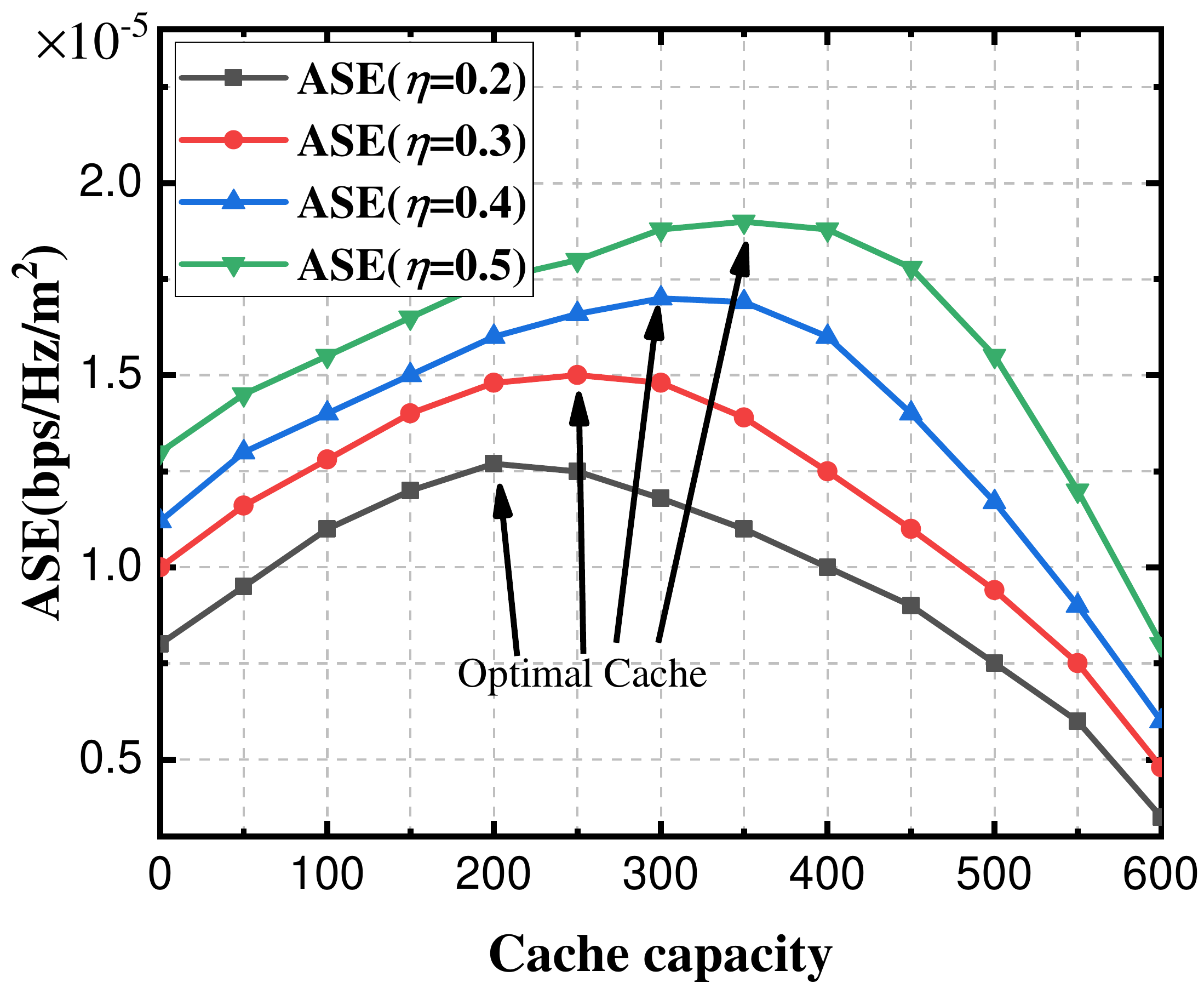}\label{ASECache}}
  \captionsetup{font={small}}
  \caption{ASE of the mABHetNets under different bandwidth partition (a) and file cache.}
  \label{EECache}
\end{figure}
To verify the impact of the bandwidth partition on ASE, we changed the bandwidth partition $\eta$ under  different the cache capacity in SBS. In Fig. \ref{ASEPartition}, we can see that, the ASE of mABHetNets first increases as the more spectrum resource is allocated to the access link including SBS and MBS. This is because the backhaul resource is enough and can be shifted to the access backhaul to improve the ASE. However, when less spectrum is used in backhaul link, the wireless backhaul link rate becomes a bottleneck and the ASE decreases. When $C = 300$, the ASE gain of caching  over not caching is about 200\%. The corresponding optimal $\eta=0.4$ is improved to 0.7 with 75\% gain.

Besides, to verify the impact of the number of the cached files, we changed the cache $C$ under  different bandwidth partition  Fig. \ref{ASECache}. ASE can be increased when the cache capacity is increased. This is because when the backhaul bandwidth is limted, the backhaul becomes the bottleneck of the file delivery. With more cache capacity, the cache hit ratio is increased, and then more files can be sent to users directly and the impact of backhaul is reduced.
However, more cache capacity consumes more power and reduce the transmission power, which reduce the ASE.

\begin{figure}[H]
  \centering
   \subfigure[]{\includegraphics[width=2.5in]{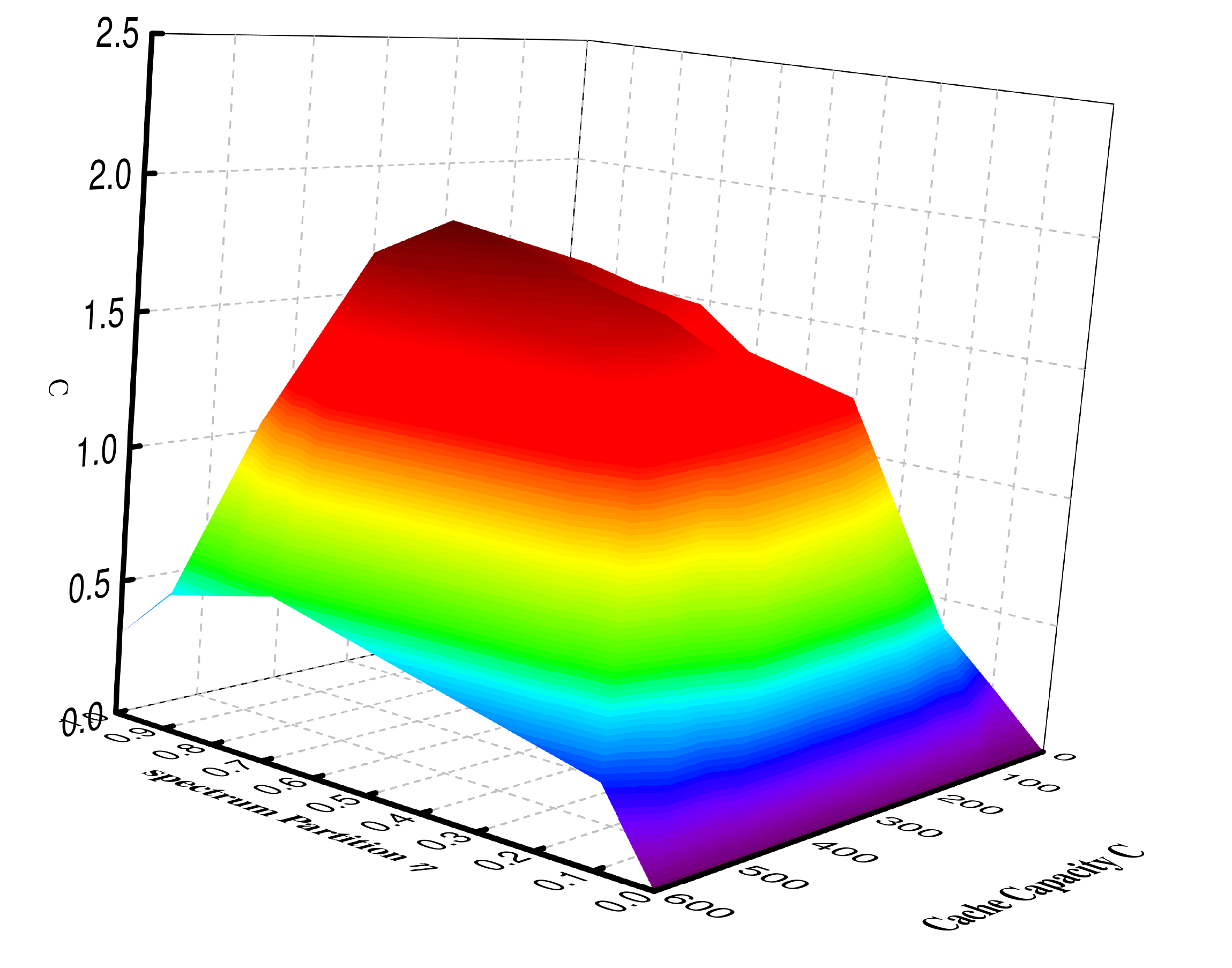}\label{ASEEtaCache2D}}
   \subfigure[]{\includegraphics[width=2.5in]{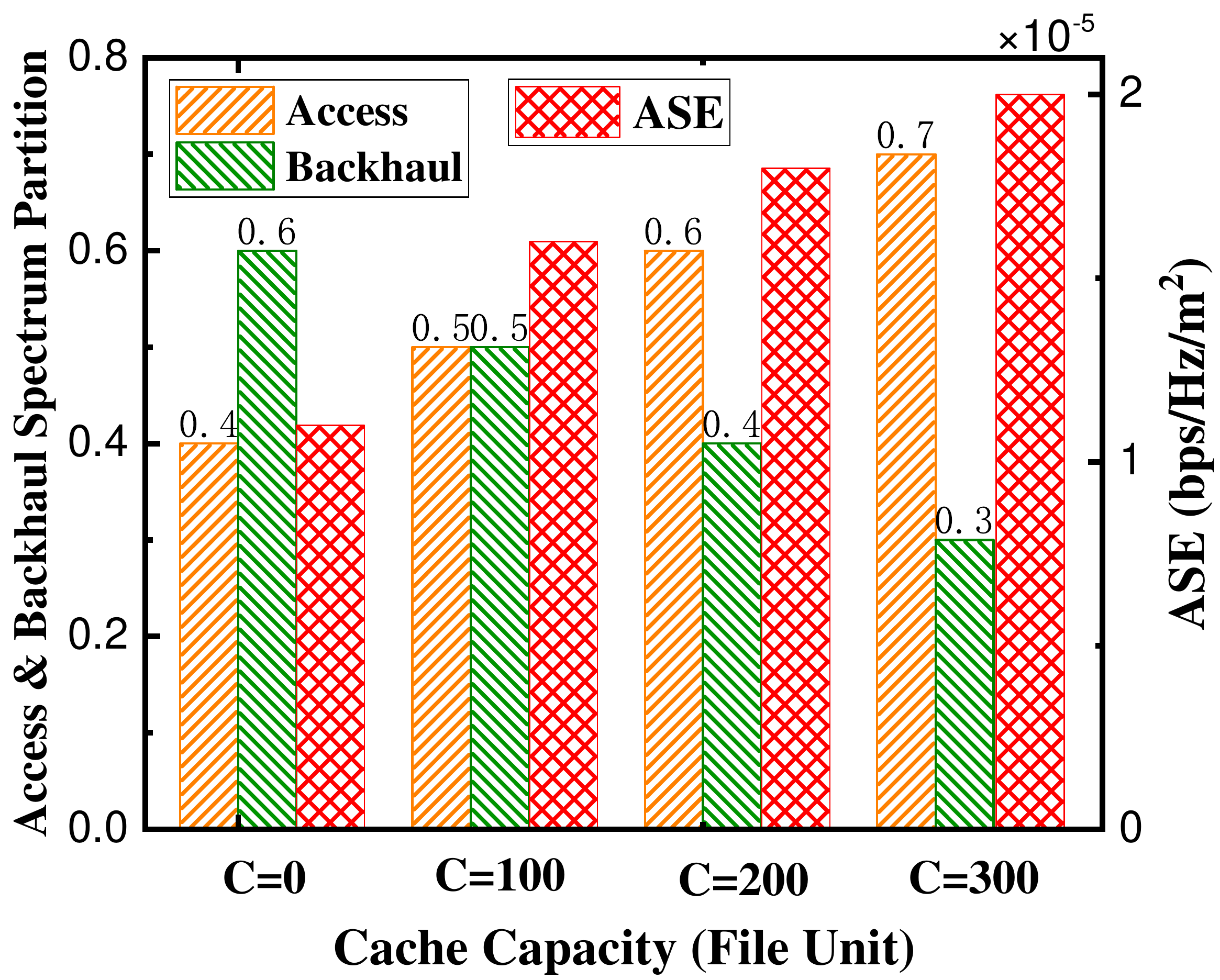}\label{ASEEtaCaDetail}}
   \captionsetup{font={small}}
  \caption{The ASE of the mABHetNets with cache capacity and bandwidth partition.}
  \label{ASEEtaCacheDetail}
\end{figure}
To further look into the joint impacts of bandwidth partition and cache capacity on the ASE, we show the 3-dimensional numerical results of ASE in Fig. \ref{ASEEtaCache2D}. Actually, there exist the optimal cache capacity and the optimal bandwidth. under some cases, cache capacity can improve ASE apparently. In Fig. \ref{ASEEtaCaDetail}, under the small cache capacity, more cache capacities cause  more access  bandwidth (over optimized bandwidth partition) and increase ASE.

\subsection{the impact of cache on $\Delta \eta$}
\begin{figure}[H]
  \centering
  \includegraphics[width=6cm]{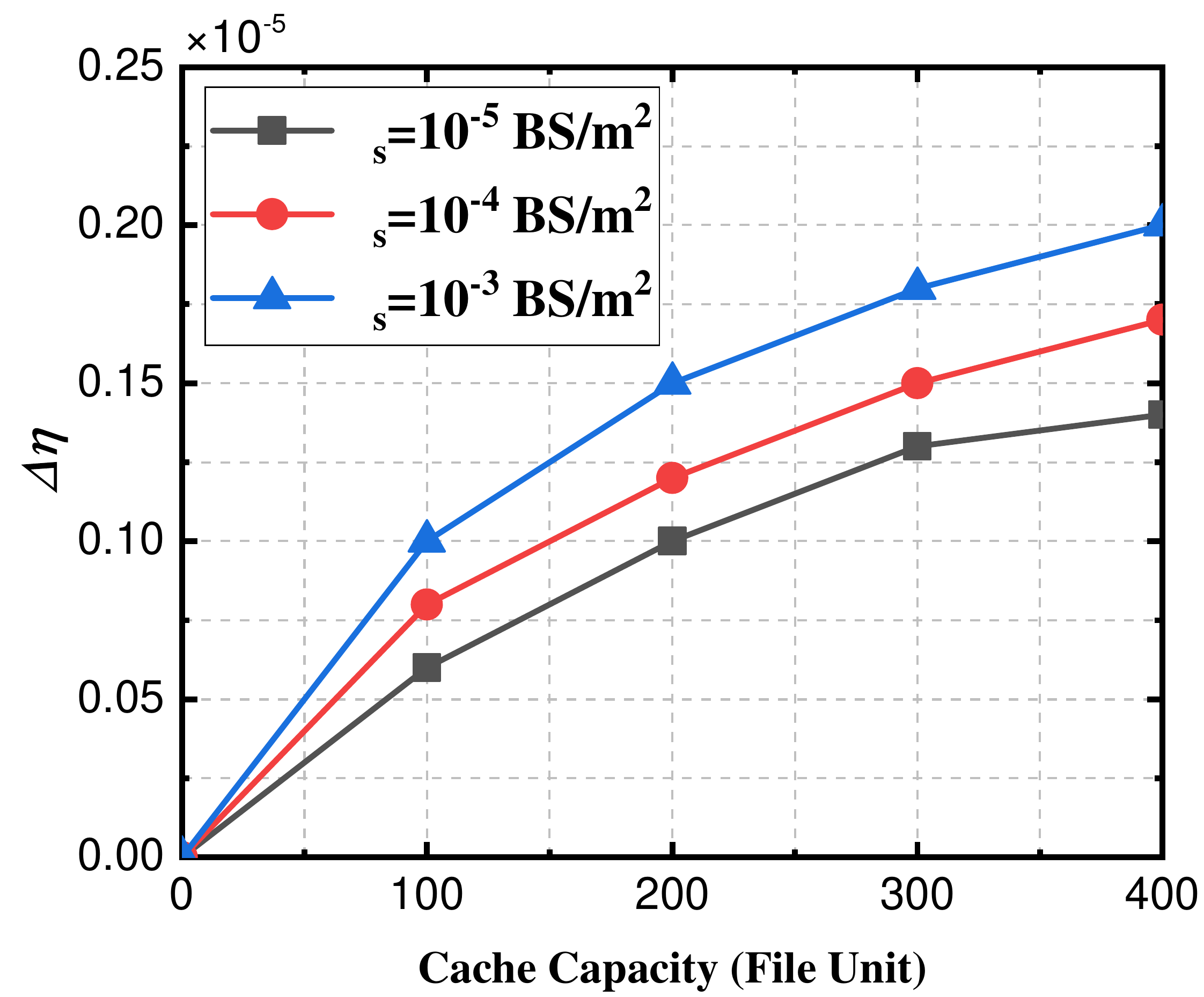}
  \caption{ Saved spectrum  for access service  under  different cache capacities.}
  \label{DeltaEtaCache}
\end{figure}
From the figure, when the cache capacity is larger, more bandwidth will be used in the access link compared with the uncached case in traditional mABHetNets. For a give cache capacity, by adjusting the bandwidth partition for access and backhaul link, the optimal ASE will be obtained. Compared with the traditional mABHetNets, more spectrum resource will be used for access link. From Fig. \ref{DeltaEtaCache}, we can see that, when 400 files are cached in SBS, over 20\% spectrum  resource is saved from the backhaul link to the access link to improve the ASE.

\subsection{Impact of Other Key Cache Parameters on ASE}
\begin{figure}[H]
  \centering
  \subfigure[]{\includegraphics[width=2.2 in]{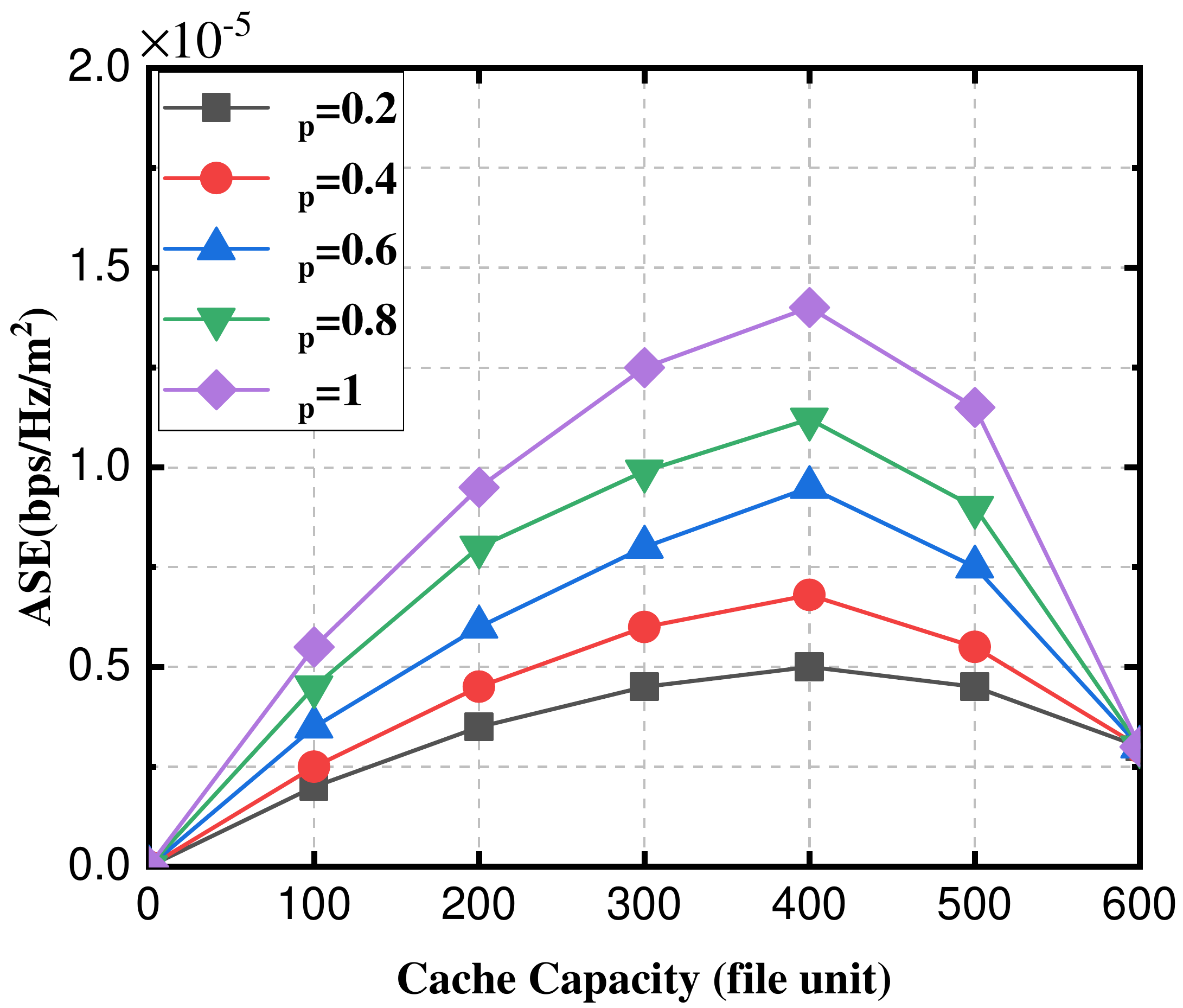}\label{ASEGammapCache}}
  \subfigure[]{\includegraphics[width=2.2 in]{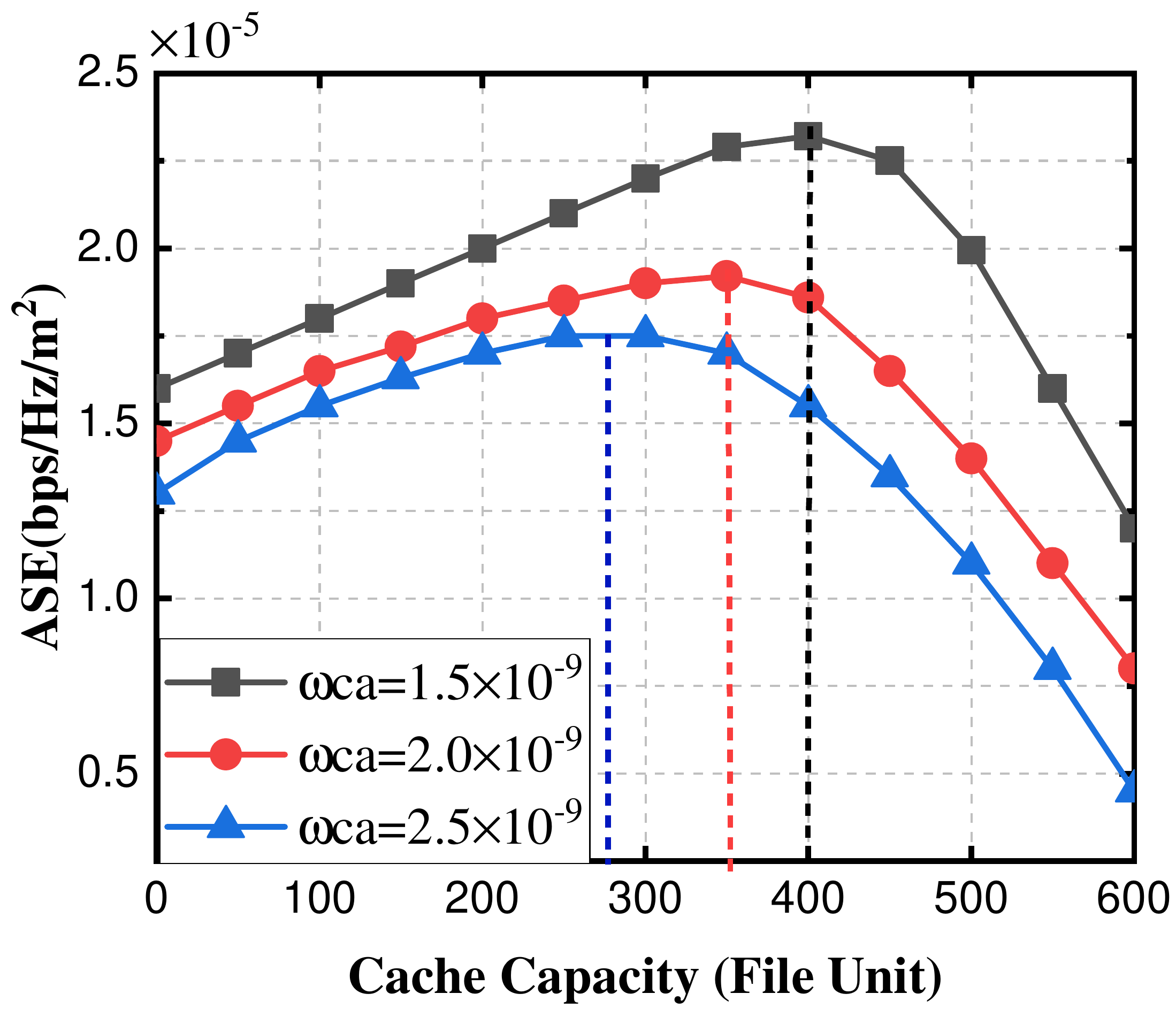}\label{ASEWcaCache}}
  \caption{Impact of cache parameter on ASE}
  \label{EECache}
\end{figure}

We want to further observe the ASE of cache-enabled mABHetNets under the key cache parameters. Both the Zipf distribution parameter $\gamma_p$ and the caching power coefficient $\omega_{ca}$ reflect the characteristics of cache-enabled mABHetNets.
In \ref{ASEGammapCache}, we show  ASE   versus the   cache capacity with different Zipf parameter $\gamma_p$. All the  results are based on the optimal bandwidth partition. We can see that the optimal cache capacity increases with increasing $\gamma_p$. With the same cache capacity, ASE increases with $\gamma_p$.    This is because the cache hit ratio $p_h$ increases with $\gamma_p$ as shown in (\ref{HitRatio}) and those cached files own a higher popularity.  The optimal AES with  optimized cache capacity in $\gamma_p$ = 1 over $\gamma_p$ = 0.2  is about 300\%.

While there are various kinds of memory technologies, we consider the three kinds that are most likely employed due to their higher power efficiencies and larger cache sizes. In Fig. \ref{ASEWcaCache}, we show the numerical results of ASE versus cache capacity under different caching power coefficients. The cache power coefficients $w_{ca}$ has an deep impact on the ASE. Low coefficient can improve the ASE over the optimized cache capacity.   This is because when the $w_{ca}$ is lower, more files can be cached with the same  power  overhead and the SBS cache hit ratio is improved. Then more files can be obtained from SBS directly and backhaul spectrum can be transferred to the  access link.

\subsection{APT and ASE}
\begin{figure}[H]
  \centering
  \includegraphics[width=6cm]{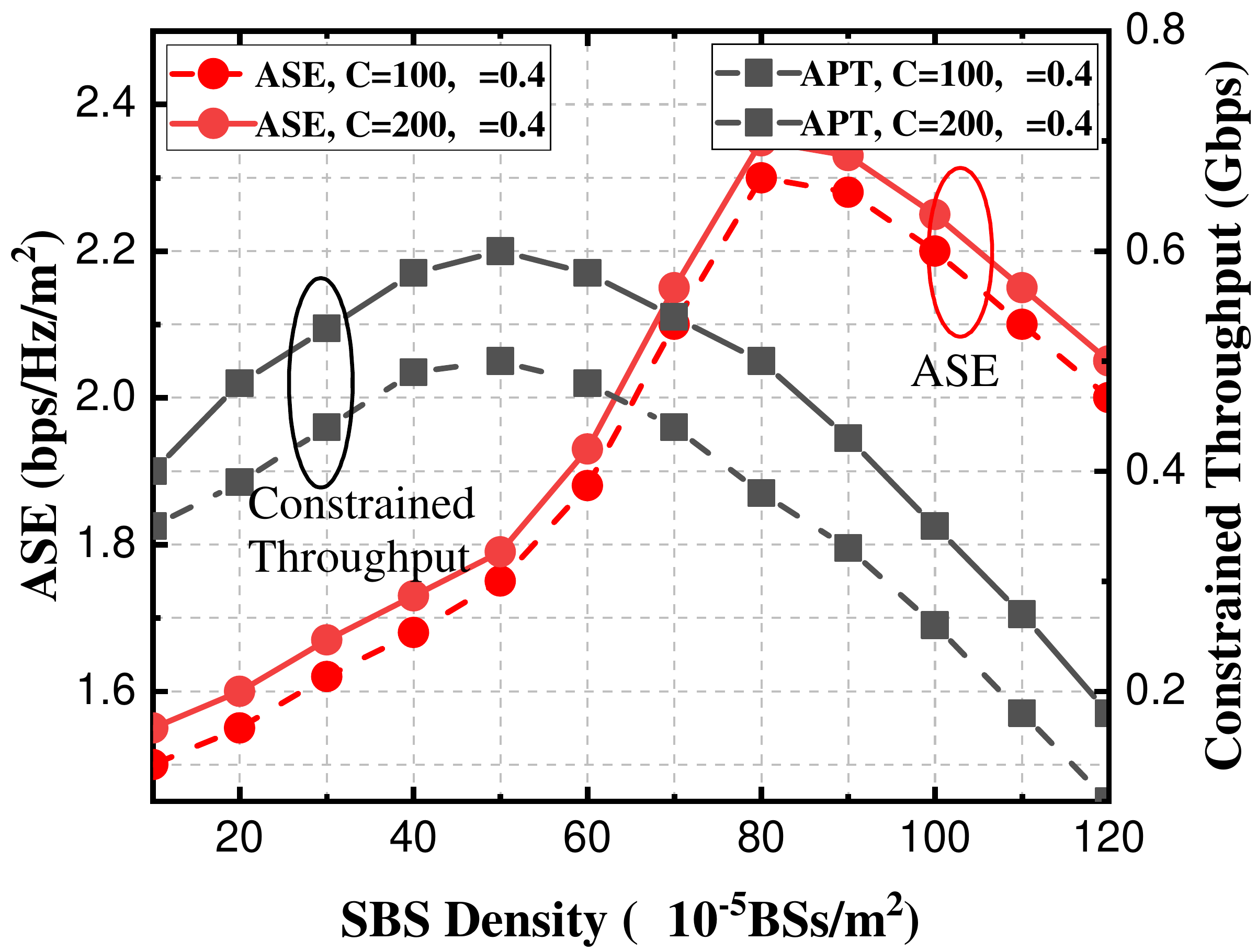}
  \caption{APT(SINR requirement $\gamma_0=5dB$) and  ASE   under  different SBS densities.}
  \label{ASEThrDen}
\end{figure}
To reflect APT and ASE under different SBS densities, we give Fig. \ref{ASEThrDen}. From the Fig.\ref{ASEThrDen}, we can see that both APT and ASE will increase with the increasing SBS density and then decrease. This is because when the when the BS density   increases, LoS transmission exists with an increasingly higher probability than NLoS transmission.  Therefore, the desired LoS signal from the associated BS is dominant. However, when the SBS becomes much
denser, the interference power are LoS dominated and the data rate decrease. Compared with ASE, APT begins to decreases at a smaller density. It is due to the fact that SINR requirement is more vulnerable to the LoS interference. When LoS interference increases with the increasing density,  it is more difficult to satisfy the SINR requirement $\gamma_0$ of APT. When the cache capacity increases a little(e.e., from 100 to 200), both APT and ASE increase. This is because that, with more cache capacity, more files are sent directly from SBS and backhaul traffic is reduced. Then, the backhaul spectrum can be shifted to access and data rate is increased.

\subsection{The Noise-limited Case and Interference-limited Case}
\begin{figure}[H]
  \centering
  \includegraphics[width=6cm]{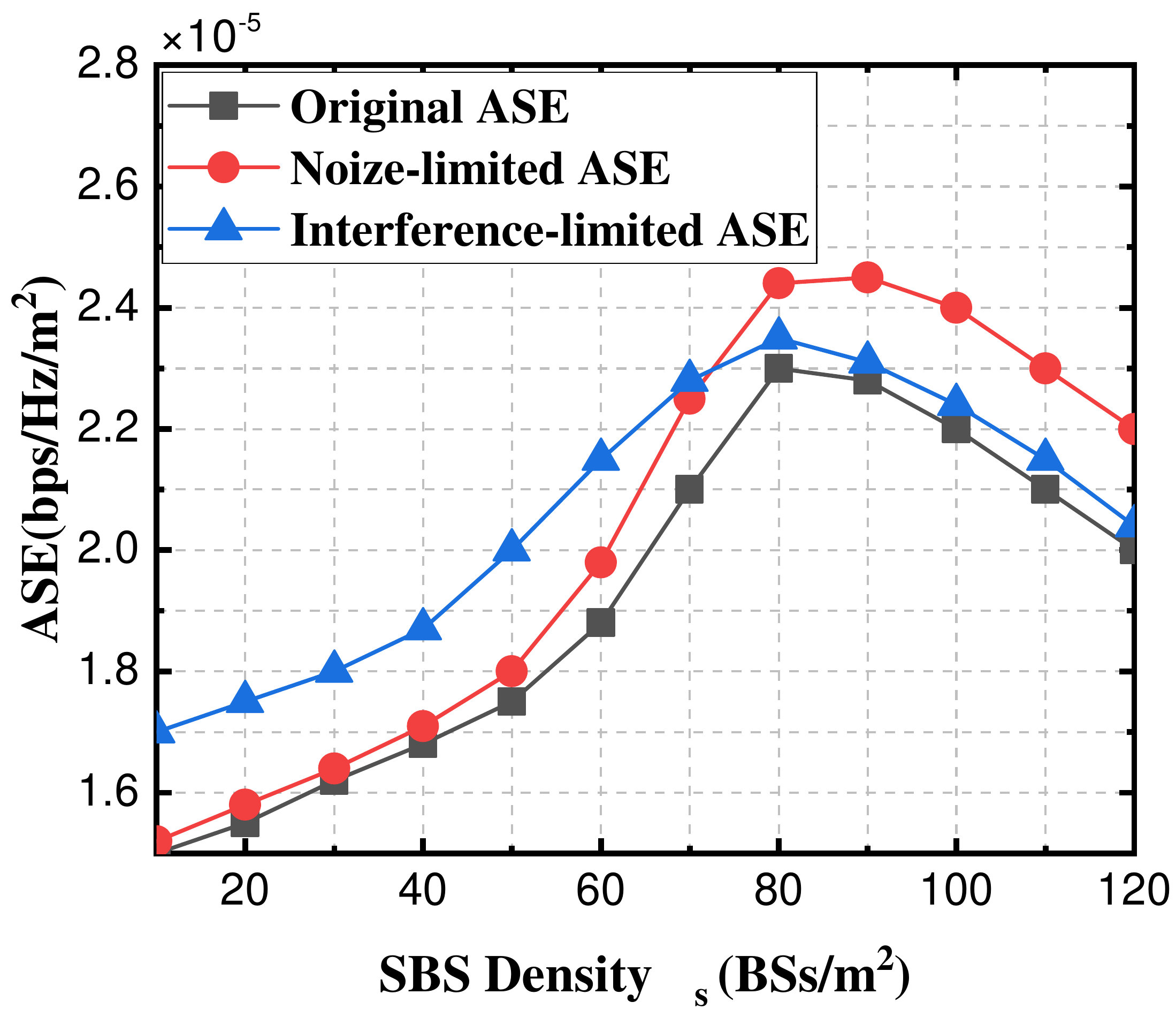}
  \captionsetup{font={small}}
  \captionsetup{font={normalsize}}
  \caption{Saved spectrum  for access service  under  different cache capacities.}
  \label{ASENoiInt}
\end{figure}
To verify the validity  of the ASE approximations  in noise-limited case  and interference-limited case, we compare the numerical results of  original ASE expression in (\ref{AreaSpectralEfficiency}) with those of noise-limited ASE (\ref{NoiseASEExpression}) and of interference-limited ASE (\ref{InterferenceASEExpression}), respectively. under the lower SBS density, the ASE expression (\ref{NoiseASEExpression}) is a more effective  upper bound of the original ASE expression in (\ref{AreaSpectralEfficiency}). That shows  the simpler expression  (\ref{NoiseASEExpression}) capture the feather that  lower density make the mmWave signal more vulnerable to the noise. Besides, in the high density case, the ASE expression (\ref{InterferenceASEExpression}) is a much closer to   the original ASE expression in (\ref{AreaSpectralEfficiency}). That results show the effectiveness of the simpler ASE expression  (\ref{NoiseASEExpression}), which mainly focus on the LoS based mmWave transmission.

\section{Appendix}

\subsection{The proof of Lemma \ref{LemmaUserClosestSBSProbabilty}}\label{Appe1:NearDistance}
We first consider the event that the    distance between the typical user  and  the nearest LoS SBS   (LoS based file transmission between the typical user and the SBS)   is $r$.
In fact,   the   event that is the joint of following two events:
The first event is  the nearest SBS of the typical user is located at distance $r$ (Event 1) and the second event is the transmission path between the typical user and the serving SBS is an LoS path (Event 2).
According to \cite{ReferUserClosestBSProbabilty}, the PDF of Event 1 with regard to $r$ is given by $\exp \left(-\pi r^{2} \lambda_{s}\right) \times 2 \pi r \lambda_{s}$. The probability of Event 2 over distance $r$  is $\mathcal{P}_L(r)$, so that we can get the PDF of the joint Event 1   and Event 2 as
\begin{align}
f_{R_{s}}^{\mathrm{L}}(r)= \mathcal{P}_{L}(r)\times \exp \left(-\pi r^{2} \lambda_{s}\right) \times 2 \pi r \lambda_{s} \label{RsLOS},
\end{align}
In a similar way, the PDF of   the event that the    distance between the typical user  and  the nearest NLoS SBS   is $r$ is
\begin{align}
f_{R_{s}}^{\mathrm{NL}}(r)=\mathcal{P}_{\mathrm{NL}}(r)   \times\exp \left(-\pi r^{2} \lambda_{s}\right) \times 2 \pi r \lambda_{s} \label{RsNLoS},
\end{align}
The PDFs of   the events that the    distance between the user  and  the nearest LoS MBS (NLoS MBS )is $r$  are
\begin{align}
f_{R_{m}}^{\mathrm{L}}(r)&=\mathcal{P}_{\mathrm{L}}(r)   \times\exp \left(-\pi r^{2} \lambda_{m}\right) \times 2 \pi r \lambda_{m} \label{RmLOS},\\
f_{R_{m}}^{\mathrm{NL}}(r)&=\mathcal{P}_{\mathrm{NL}}(r)   \times\exp \left(-\pi r^{2} \lambda_{m}\right) \times 2 \pi r \lambda_{m} \label{RmNLoS}
\end{align}

\subsection{The proof of Lemma \ref{LemmaUserAssociatedSBS}}\label{Appen2:AssoPro}
The typical user may associate with the SBS tier by either LoS channel or NLoS channel.
We derive the probability of the first event that the user is associated with the SBS by the wireless LoS  link. Such event has three cases: the interference from the NLoS SBS, the interference from the  LoS MBS and the interference from the NLoS MBS. Therefore, the probability that the user obtain the desired LoS signal from the SBS is $F_{s}^{L}(r)=p_{ln}^{ss}(r)p_{ll}^{sm}(r)p_{ln}^{sm}(r)f_{R_s}^{\mathrm{L}}(r)$ where
\begin{enumerate}
    \item   The user is associated with the LoS SBS and the interference is from NLoS SBS.
            \begin{align}\label{LSBSNLSBS}
             p_{ln}^{ss}(r)
             &=\mathbb{P}[P_{s}^{tr}  B_sh_{s}  A_\mathrm{L} r^{-\alpha_L} \geq P_{s}^{tr}  B_sh_{s} A_{\mathrm{NL}} r_{s}^{-\alpha_{\mathrm{NL}}}]\\
             &=\mathbb{P}\left[ r_{s}\geq\left(\frac{A_\mathrm{L}}{A_\mathrm{NL}}\right)^{\frac{-1}{\alpha_{NL}}} r^{\frac{\alpha_{L}}{\alpha_{NL}}} \right]
             =e^{-\lambda_s\pi\left[ \left(\frac{A_\mathrm{L}}{A_\mathrm{NL}}\right)^{\frac{-1}{\alpha_{NL}}} r^{\frac{\alpha_{L}}{\alpha_{NL}}}\right]^{2}}
            \end{align}
            where the last step is based on the derivation in \cite{ReferUserClosestBSProbabilty}.
    \item   The user is associated with the LoS SBS and the interference is from LoS MBS.
    \begin{align}\label{LSBSLMBS}
             p_{ll}^{sm}(r)
             &=\mathbb{P}[P_{s}^{tr} h_{s}  A_\mathrm{L} r^{-\alpha_L} \geq P_{m}^{tr} B_{m}h_m A_{\mathrm{L}} r_{m}^{-\alpha_L}]\\
             &=\mathbb{P}\left[ r_{m}\geq\left(\frac{P_s^{tr}B_{s} h_s}{P_m^{tr}B_{m} h_m}\right)^{\frac{-1}{\alpha_{L}}} r \right]
             =e^{-\lambda_m\pi\left[\left(\frac{P_s^{tr} B_{s}h_s}{P_m^{tr} B_{m} h_m}\right)^{\frac{-1}{\alpha_{L}}} r\right]^{2}}
      \end{align}
    \item   The user is associated with the LoS SBS and the interference is from NLoS MBS.
      \begin{align}\label{LSBSLMBS}
             p_{ln}^{sm}(r)
             &=\mathbb{P}[P_{s}^{tr}B_{s} h_{s}  A_\mathrm{L} r^{-\alpha_L} \geq P_{m}^{tr} B_{m}h_{m} A_{\mathrm{NL}} r_{m}^{-\alpha_{\mathrm{NL}}}]\\
             &=\mathbb{P}\left[ r_{m}\geq\left(\frac{P_s^{tr} B_{s}h_s A_\mathrm{L}}{P_m^{tr} B_{m}h_m A_\mathrm{NL}}\right)^{\frac{-1}{\alpha_{\mathrm{NL}}}} r^{\frac{\alpha_{\mathrm{L}}}{\alpha_{\mathrm{NL}}}} \right]
             =e^{-\lambda_m\pi\left[\left(\frac{P_s^{tr} B_{s}h_s A_\mathrm{L}}{P_m^{tr} B_{m}h_m A_\mathrm{NL}}\right)^{\frac{-1}{\alpha_{\mathrm{NL}}}} r^\frac{\alpha_{\mathrm{L}}}{\alpha_{\mathrm{NL}}}\right]^{2}}
      \end{align}
\end{enumerate}

Then,  the probability that the user obtain the desired NLoS signal from the SBS is $F_{s}^{NL}(r)=p_{nl}^{ss}(r)$ $ p_{nl}^{sm}(r) p_{nn}^{sm}(r)f_{R_s}^{\mathrm{NL}}(r)$ where
\begin{enumerate}
    \item   The user is associated with the NLoS SBS and the interference is from LoS SBS.
            \begin{align}\label{NLSBSLSBS}
             p_{nl}^{ss}(r)&= P[P_s^{tr}B_{s}h_sA_{NL}r^{-\alpha_{NL}}\geq P_s^{tr}B_{s}h_sA_{L}r_s^{-\alpha_L}]
             =e^{-\lambda_s\pi\left[ \left(\frac{A_\mathrm{NL}}{A_\mathrm{L}}\right)^{\frac{-1}{\alpha_{L}}} r^{\frac{\alpha_{NL}}{\alpha_{L}}}\right]^{2}}
            \end{align}
    \item   The user is associated with the NLoS SBS and the interference is from LoS MBS.
    \begin{align}\label{NLSBSLMBS}
             p_{nl}^{sm}(r)&=P[P_s^{tr}B_{s}h_sA_{NL}r^{-\alpha_{NL}}\geq P_m^{tr}B_{m}h_mA_Lr_m^{-\alpha_L}]
             =e^{-\lambda_m\pi\left[\left(\frac{P_s^{tr} B_{s}h_s A_\mathrm{NL}}{P_m^{tr} B_{m}h_m A_{\mathrm{L}}}\right)^{\frac{-1}{\alpha_\mathrm{L}}} r^\frac{\alpha_\mathrm{NL}}{\alpha_\mathrm{L}}\right]^{2}}
      \end{align}
    \item   The user is associated with the NLoS SBS and the interference is from NLoS MBS.
      \begin{align}\label{NLSBSNLMBS}
             p_{nn}^{sm}(r)&=P[P_s^{tr}B_{s}h_sA_{NL}r^{-\alpha_{NL}}\geq P_m^{tr}B_{m}h_mA_{NL}r_m^{-\alpha_{NL}}]
             =e^{-\lambda_m\pi\left[\left(\frac{P_s^{tr}B_{s} h_s}{P_m^{tr} B_{m} h_m}\right)^{\frac{-1}{\alpha_{\mathrm{NL}}}} r\right]^{2}}
      \end{align}
\end{enumerate}

Then,  the probability that the user obtain the desired LoS signal from the MBS is $F_{m}^{L}(r)=p_{1n}^{mm}(r)$ $ p_{ll}^{ms}(r) p_{ln}^{sm}(r) f_{R_m}^{\mathrm{L}}(r)$ where
\begin{enumerate}
    \item   The user is associated with the LoS MBS and the interference is from NLoS MBS.
            \begin{align}\label{LMBSNLMBS}
             p_{ln}^{mm}(r)=P[P_m^{tr}B_{m}h_mA_Lr^{-\alpha_L}\geq P_m^{tr}B_{m}h_mA_{NL}r_m^{-\alpha_{NL}}] =e^{-\lambda_m\pi\left[\left(\frac{A_\mathrm{L}}{A_\mathrm{NL}}\right)^{\frac{-1}{\alpha_{NL}}} r^{\frac{\alpha_{L}}{\alpha_{NL}}}\right]^{2}}
            \end{align}
    \item   The user is associated with the LoS MBS and the interference is from LoS SBS.
    \begin{align}\label{LSBSLMBS}
             p_{ll}^{ms}(r) =p[P_m^{tr}B_{m}h_mA_Lr^{-\alpha_L}\geq P_s^{tr}B_{s}h_sA_{L}r_s^{-\alpha_{L}}]=e^{-\lambda_s\pi\left[\left(\frac{P_m^{tr}B_{m} h_m}{P_s^{tr} B_{s}h_s}\right)^{\frac{-1}{\alpha_{L}}} r\right]^{2}}
      \end{align}
    \item   The user is associated with the LoS MBS and the interference is from NLoS SBS.
      \begin{align}\label{LMBSNLSBS}
             p_{ln}^{sm}(r) =P[P_m^{tr}B_{m}h_mA_Lr^{-\alpha_L}\geq P_s^{tr}B_{s}h_sA_{NL}r_s^{-\alpha_{NL}}]=e^{-\lambda_s\pi\left[\left(\frac{P_m^{tr}B_{m} h_m A_\mathrm{L}}{P_s^{tr} B_{s}h_s A_\mathrm{NL}}\right)^{\frac{-1}{\alpha_{\mathrm{NL}}}} r^{\frac{\alpha_{\mathrm{L}}}{\alpha_{\mathrm{NL}}}}\right]^{2}}
      \end{align}
\end{enumerate}

Then,  the probability that the user obtain the desired NLoS signal from the MBS is $F_{m}^{NL}(r)=p_{nl}^{mm}(r)$ $ p_{nl}^{ms}(r) p_{nn}^{ms}(r) f_{R_m}^{\mathrm{NL}}(r)$ where

\begin{enumerate}
    \item   The user is associated with the NLoS MBS and the interference is from LoS MBS.
            \begin{align}\label{LMBSNLMBS}
             p_{nl}^{mm}(r)=P[P_m^{tr}B_{m}h_mA_{NL}r^{-\alpha_{NL}}\geq P_m^{tr}B_{m}h_mA_Lr_m^{-\alpha_L}]
             =e^{-\lambda_m\pi\left[\left(\frac{A_\mathrm{NL}}{A_\mathrm{L}}\right)^{\frac{-1}{\alpha_{L}}} r^{\frac{\alpha_{\mathrm{NL}}}{\alpha_{\mathrm{L}}}}\right]^{2}}
            \end{align}
    \item   The user is associated with the NLoS MBS and the interference is from LoS SBS.
    \begin{align}\label{LSBSLMBS}
             p_{nl}^{ms}(r)=P[P_m^{tr}B_{m}h_mA_{NL}r^{-\alpha_{NL}}\geq P_s^{tr}B_{s}h_sA_Lr_s^{-\alpha_L}]
             =e^{-\lambda_s\pi\left[\left(\frac{P_m^{tr} B_{m}h_mA_{NL}}{P_s^{tr} B_{s} h_sA_{L}}\right)^{\frac{-1}{\alpha_{L}}} r^{\frac{ \alpha_\mathrm{NL}}{\alpha_{L}}}\right]^{2}}
      \end{align}
    \item   The user is associated with the NLoS MBS and the interference is from NLoS SBS.
      \begin{align}\label{LMBSNLSBS}
             p_{nn}^{ms}(r)=P[P_m^{tr}B_{m}h_mA_{NL}r^{-\alpha_{NL}}\geq P_s^{tr}B_{s}h_sA_{NL}r_s^{-\alpha_{NL}}]
             =e^{-\lambda_s\pi\left[\left(\frac{P_m^{tr} B_{m}h_m}{P_s^{tr}B_{s} h_s}\right)^{\frac{-1}{\alpha_{\mathrm{NL}}}} r \right]^{2}}
      \end{align}
\end{enumerate}

The probability that   SBS obtains the desired LoS signal from   MBS is $\small F_{bh}^{L}(r)=p_{ln}^{bh}(r) f_{R_{bh}}^{\mathrm{L}}(r)$ where $p_{ln}^{bh}(r)=P[P_m^{tr}B_{m}h_mA_Lr^{-\alpha_L}\geq P_m^{tr}B_{m}h_mA_{NL}r_{bh}^{-\alpha_{NL}}] =e^{-\lambda_m\pi\left[\left(\frac{A_\mathrm{L}}{A_\mathrm{NL}}\right)^{\frac{-1}{\alpha_{NL}}} r^{\frac{\alpha_{L}}{\alpha_{NL}}}\right]^{2}}$ is the probability that the SBS is associated with the LoS MBS and the interference is from NLoS MBS.

The probability that   SBS obtains the desired NLoS signal from   MBS is $F_{bh}^{NL}(r)=p_{nl}^{bh}(r) f_{R_bh}^{\mathrm{NL}}(r)$ where $\footnotesize
             p_{nl}^{bh}(r)=P[P_m^{tr}B_{m}h_mA_{NL}r^{-\alpha_{NL}}\geq P_m^{tr}B_{m}h_mA_Lr_{bh}^{-\alpha_L}]
             =e^{-\lambda_m\pi\left[\left(\frac{A_\mathrm{NL}}{A_\mathrm{L}}\right)^{\frac{-1}{\alpha_{L}}} r^{\frac{\alpha_{\mathrm{NL}}}{\alpha_{\mathrm{L}}}}\right]^{2}}
            $ is the probability that the SBS is associated with the NLoS MBS and the interference is from LoS MBS.

\subsection{Proof of Proposition \ref{ProUserSINRCoverageProbability}}\label{Appen3:ProUserSINRCoverageProbability}

Then we first focus on the SINR distribution of  a user   covered by SBS:
\begin{align}
  P_{s}^{cov}(\gamma)=P_{s,L}^{cov}(\gamma)+P_{s,NL}^{cov}(\gamma)
  \end{align}
where the SINR distribution of  a user   covered by LoS  SBS:
\begin{align}
  &P_{s,L}^{cov}(\gamma)=\mathbb{E}_{r }\left[\mathbb{P}\left[\mathrm{SINR}_{s}^{\mathrm{L}} (r ) \geq \gamma\right]\right]=\int_{0}^{\infty} \mathbb{P}\left[\operatorname{SINR}_{s}^{\mathrm{L}}(r)>\gamma\right] F_{s}^{\mathrm{L}}(r)\mathrm{d}r
\end{align}
and the SINR distribution of  a user   covered by NLoS  SBS:
\begin{align}
  &P_{s,NL}^{cov}(\gamma)=\mathbb{E}_{r }\left[\mathbb{P}\left[\mathrm{SINR}_{s}^{\mathrm{NL}} (r ) \geq \gamma\right]\right]=\int_{0}^{\infty} \mathbb{P}\left[\operatorname{SINR}_{s}^{\mathrm{NL}}(r)>\gamma\right] F_{s}^{\mathrm{NL}}(r)\mathrm{d}r
\end{align}
where $\gamma$ is the threshold for successful demodulation and decoding at the receiver. $\mathbb{P}\left[\mathrm{SINR}_{s}^{\mathrm{L}} (r ) \geq \gamma\right]$ means the probability of the event that the SINR of  the user covered by SBS is over $\gamma$  via the LoS path at distance $r$:
\begin{align}
&\mathbb{P}\left[\operatorname{SINR}_{s}^{\mathrm{L}}(r)  \geq \gamma\right]=\mathbb{P}\left[\frac{P_{s}^{tr}B_{s} A_{\mathrm{L}} r^{-\alpha_{\mathrm{L}}}}{I_s+I_m+N_{0}} \geq \gamma\right]\nonumber\\
&=\mathbb{P}\left[h_{m } \geq \frac{\gamma\left(I_s+I_m+N_{0}\right)}{P_{s}^{tr} B_{s}A_{\mathrm{L}} r^{-\alpha_{\mathrm{L}}} }\right]\stackrel{(a)}{=} \exp \left(\frac{-\gamma N_{0}}{P_{s}^{tr} B_{s}A_{\mathrm{L}} r^{-\alpha_{\mathrm{L}}} }\right)  \mathcal{L}_{I_{s,m}}^{\mathrm{L}}\left(\gamma r^{\alpha_{\mathrm{L}}} \right)
\end{align}
Besides, $\mathbb{P}\left[\mathrm{SINR}_{s}^{\mathrm{NL}} (r ) \geq \gamma\right]$ means the probability of the event that the SINR of  the user covered by SBS is over $\gamma$  via the NLoS path at distance $r$:
\begin{align}
&\mathbb{P}\left[\operatorname{SINR}_{s}^{\mathrm{NL}}(r)  \geq \gamma\right]=\mathbb{P}\left[\frac{P_{s}^{tr} B_{s} G_{s} A_{\mathrm{L}} r^{-\alpha_{\mathrm{NL}}}}{I_s+I_m+N_{0}} \geq \gamma\right]\nonumber\\
&=\mathbb{P}\left[h_{m0} \geq \frac{\gamma\left(I_s+I_m+N_{0}\right)}{P_{s}^{tr} B_{s}A_{\mathrm{NL}} r^{-\alpha_{\mathrm{NL}}} }\right]\stackrel{(a)}{=} \exp \left(\frac{-\gamma N_{0}}{P_{s}^{tr} B_{s} h_{s}A_{\mathrm{NL}} r^{-\alpha_{\mathrm{NL}}} }\right)  \mathcal{L}_{I_{s,m}}^{\mathrm{NL}}\left(\gamma r^{\alpha_{\mathrm{NL}}} \right)
\end{align}
where (a) follows from small fading $h$$\sim$$\exp(1)$. Here the Rayleigh fading is considered.    $\mathcal{L}_{I_{s,m}}$  is the Laplace transform of the cumulative interference from  the SBS tier.
\begin{small}
\begin{align}\label{LoSInterferenceFromtheSBS}
   &\mathcal{L}_{I_{s,m}}^{\mathrm{L}}\left( \gamma r^{\alpha_{\mathrm{L}}} \right)\\
   &\stackrel{(b)}{=}
   \exp \left(-2 \pi  \lambda_{s}\left(\int_{r}^{\infty}  \frac{\mathcal{P}_L(u)u}{1+\frac{P_s^{tr}B_{s}h_sA_Lr^{-\alpha_{\mathrm{L}}}}{\gamma P_s^{tr}B_{s}h_sA_Lu^{-\alpha_{\mathrm{L}}}}} d u+\int_{\left(\frac{A^{\mathrm{L}}}{A^{\mathrm{NL}}}\right)^{\frac{-1}{\alpha^{\mathrm{NL}}}} r^{\frac{\alpha_{\mathrm{L}}}{\alpha_{\mathrm{NL}}}}}^{\infty}  \frac{\mathcal{P}_{\mathrm{NL}}(u)u}{1+\frac{  P_s^{tr}B_{s}h_sA_{\mathrm{L}}r^{-\alpha_{\mathrm{L}}}}{\gamma  P_s^{tr}B_{s}h_s A_{\mathrm{NL}}u^{-\alpha_{\mathrm{NL}}}}} d u\right)\right)\nonumber\\
   &\times \exp \left(-2 \pi  \lambda_{m}\left(\int_{\left(d_1\right)^{\frac{-1}{\alpha_{L}}} r}^{\infty} \frac{\mathcal{P}_{L}(u)u}{1+\frac{P_{s}^{tr}B_{s}h_sA_L r^{-\alpha_{\mathrm{L}}}}{\gamma P_{m}^{tr} B_{m}h_m A_L u^{-\alpha_{\mathrm{L}}}}} d u+\int_{\left(d_2\right)^{\frac{-1}{\alpha_{\mathrm{NL}}}} r^\frac{\alpha_{\mathrm{L}}}{\alpha_{\mathrm{NL}}}}^{\infty}  \frac{\mathcal{P}_{\mathrm{NL}}(u)u}{1+\frac{P_{s}^{tr} B_{s}h_sA_L r^{-\alpha_{\mathrm{L}}}}{\gamma P_{m}^{tr}B_{m}h_mA_{L} u^{-\alpha_{\mathrm{L}}}}} d u\right)\right)\nonumber
\end{align}
\end{small}
where step (b) is based on \cite{ReferUserClosestBSProbabilty}. $d_1=\frac{P_s^{tr}B_{s} h_s}{P_m^{tr}B_{m} h_m}$ and $d_2=\frac{P_s^{tr} B_{s}h_s A_\mathrm{L}}{P_m^{tr}B_{m} h_m A_\mathrm{NL}}$
Following the same logic,
$  \mathcal{L}_{I_{s,m}}^{\mathrm{NL}}\left(\gamma r^{\alpha_{\mathrm{NL}}} \right),$
$  \mathcal{L}_{I'_{s,m}}^{\mathrm{L}}\left(\gamma r^{\alpha_{\mathrm{L}}} \right),$
$  \mathcal{L}_{I'_{s,m}}^{\mathrm{NL}}\left(\gamma r^{\alpha_{\mathrm{NL}}} \right),$
$  \mathcal{L}_{I_{bh}}^{\mathrm{L}}\left(\gamma r^{\alpha_{\mathrm{L}}} \right),$
$  \mathcal{L}_{I_{bh}}^{\mathrm{NL}}\left(\gamma r^{\alpha_{\mathrm{NL}}} \right)$
can be also obtained.

In the next, we  focus on the SINR distribution of  a user   covered by  MBS :
\begin{align}
  P_{m}^{cov}
  &=P_{m,L}^{cov}(\gamma)+P_{m,NL}^{cov}(\gamma)
  =\mathbb{E}_{r }\left[\mathbb{P}\left[\mathrm{SINR}_{m}^{\mathrm{L}} (r ) \geq \gamma\right]\right]+\mathbb{E}_{r }\left[\mathbb{P}\left[\mathrm{SINR}_{m}^{\mathrm{NL}} (r ) \geq \gamma\right]\right]\nonumber\\
  &=\int_{0}^{\infty} \mathbb{P}\left[\operatorname{SINR}_{m}^{\mathrm{L}}(r)>\gamma\right] F_{m}^{\mathrm{L}}(r) \mathrm{d} r
  +\int_{0}^{\infty}\mathbb{P}\left[\operatorname{SINR}_{m}^{\mathrm{NL}}(r)>\gamma\right] F_{m}^{\mathrm{NL}}(r) \mathrm{d} r\nonumber
\end{align}
where
$
\mathbb{P}\left[\operatorname{SINR}_{m}^{\mathrm{L}}(r)  \geq \gamma\right] = \exp \left(\frac{-\gamma N_{0}}{P_{m}^{tr} B_{m}A_{\mathrm{L}} r^{-\alpha_{\mathrm{L}}} }\right)  \mathcal{L}_{I_{s,m}^{'}}^{\mathrm{L}}\left(\gamma r^{\alpha_{\mathrm{L}}} \right)
$
and
$
\mathbb{P}\left[\operatorname{SINR}_{m}^{\mathrm{NL}}(r)  \geq \gamma\right]\\=  \exp \left(\frac{-\gamma N_{0}}{P_{m}^{tr} B_{m} g_{m}A_{\mathrm{NL}} r^{-\alpha_{\mathrm{NL}}} }\right)  \mathcal{L}_{I_{s,m}^{'}}^{\mathrm{NL}}\left(\gamma r^{\alpha_{\mathrm{NL}}} \right)
$

Considering SBS is also covered by MBS via the wireless backhaul link, we  focus on the SINR distribution of  a SBS is covered by MBS:
\begin{align}
  P_{bh}^{cov}(\gamma)
  &=P_{bh,L}^{cov}(\gamma)+P_{bh,NL}^{cov}(\gamma)\\
  &=\mathbb{E}_{r }\left[\mathbb{P}\left[\mathrm{SINR}_{bh}^{\mathrm{L}} (r ) \geq \gamma\right]\right]+\mathbb{E}_{r }\left[\mathbb{P}\left[\mathrm{SINR}_{bh}^{\mathrm{NL}} (r ) \geq \gamma\right]\right]\nonumber\\
  &=\int_{0}^{\infty} \mathbb{P}\left[\operatorname{SINR}_{bh}^{\mathrm{L}}(r)>\gamma\right] F_{bh}^{\mathrm{L}}(r)\mathrm{d} r+\int_{0}^{\infty}\mathbb{P}\left[\operatorname{SINR}_{bh}^{\mathrm{NL}}(r)>\gamma\right] F_{bh}^{\mathrm{NL}}(r) \mathrm{d} r\nonumber
\end{align}
where
$
\mathbb{P}\left[\operatorname{SINR}_{bh}^{\mathrm{L}}(r)  \geq \gamma\right] = \exp \left(\frac{-\gamma N_{0}}{P_{m}^{tr} B_{m} g_{m}A_{\mathrm{L}} r^{-\alpha_{\mathrm{L}}} }\right)  \mathcal{L}_{I_{bh}}^{\mathrm{L}}\left(\gamma r^{\alpha_{\mathrm{L}}} \right)
$
and
$
\mathbb{P}\left[\operatorname{SINR}_{bh}^{\mathrm{NL}}(r)  \geq \gamma\right]\\=  \exp \left(\frac{-\gamma N_{0}}{P_{m}^{tr} B_{m} g_{m}A_{\mathrm{NL}} r^{-\alpha_{\mathrm{NL}}} }\right)  \mathcal{L}_{I_{bh}}^{\mathrm{NL}}\left(\gamma r^{\alpha_{\mathrm{NL}}} \right)
$

\subsection{Proof of Proposition \ref{ProRateCoverageDistribution}} \label{Appen4:ProRateCoverageDistribution}
For the user associating with SBS tier,  the user spectral efficiency   distribution is $\mathbb{P}\left[R_s>\rho\right]$
 where $R_s$ denotes the spectral efficiency of a typical user associating a serving SBS and $\rho$ is the spectral efficiency  requirement. When a user associated with the SBS is requesting  files, the cached files will be delivered by the SBS directly and the uncached files  will be delivered to the user through the wireless backhaul link and wireless access link. Therefore, the   spectral efficiency of the user associated with serving SBS is not only  limited by the wireless link capacity and backhaul link capacity, but also is related with the cache hit ratio $p_h$. Namely, $R_s=\min \{ \eta\log_2(1+\mathrm{SINR}_{s}(r)),\frac{ 1-\eta }{1-p_h}\log_2(1+\mathrm{SINR}_{bh}(r))\}$.
 However, since the transmission of the wireless access link and the wireless backhaul link  is either LoS or NloS, the below four cases of $\mathbb{P}\left[R_s>\rho \right]$ are analyzed:
\begin{itemize}
  \item  Files are delivered by both  LoS based wireless access link and the wireless backhaul link.
        \begin{align}
        &\mathbb{P}\left[R_{s,1}>\rho \right] \\
        &=\mathbb{P}\left[\min \{ \eta\log_2(1+\mathrm{SINR}_{s}^{\mathrm{L}}(r)),\frac{ 1-\eta }{ 1-p_h }\log_2(1+\mathrm{SINR}_{bh}^{\mathrm{L}}(r))\}\geq \rho\right]\nonumber\\
        &=\mathbb{P}\left[\eta  \log _{2}(1+\operatorname{SINR}_{s}^{\mathrm{L}} (r))>\rho\right]\times\mathbb{P}\left[(1-\eta) \log _{2}(1+\operatorname{SINR}_{bh}^{\mathrm{L}} (r))>(1-p_h)\rho\right]\nonumber\\
        &=\mathbb{P}\left[\mathrm{SINR}_{s}^{\mathrm{L}}(r)\geq 2^{\frac{\rho}{\eta }}-1\right]\times\mathbb{P}\left[\mathrm{SINR}_{bh}^{\mathrm{L}}(r)\geq 2^{\frac{(1-p_h)\rho}{(1-\eta) }}-1\right]\nonumber\\
        &=  \mathbb{P}_{s,L}^{cov}(2^{\frac{\rho}{\eta}}-1|r_{s} ) \times  \mathbb{P}_{bh,L}^{cov}(2^{\frac{(1-p_h)\rho}{(1-\eta) }}-1|r_{bh} )\nonumber
    \end{align}
  \item Files are delivered by    LoS  based  wireless access link and  NLoS based wireless backhaul.
        \begin{align}
        &\mathbb{P}\left[R_{s,2}>\rho \right]=\mathbb{P}_{s,\mathrm{L}}^{cov}(2^{\frac{\rho}{\eta}}-1|r_{s} ) \times  \mathbb{P}_{bh,\mathrm{NL}}^{cov}(2^{\frac{(1-p_h)\rho}{(1-\eta) }}-1|r_{bh} )
        \end{align}
  \item Files are delivered by    NLoS  based  wireless access link and  LoS based wireless backhaul.
        \begin{align}
        &\mathbb{P}\left[R_{s,3}>\rho \right]=\mathbb{P}_{s,\mathrm{NL}}^{cov}(2^{\frac{\rho}{\eta}}-1|r_{s} ) \times  \mathbb{P}_{bh,\mathrm{L}}^{cov}(2^{\frac{(1-p_h)\rho}{(1-\eta) }}-1|r_{bh} )
    \end{align}
  \item Files are delivered by    NLoS  based  wireless access link and  NLoS based wireless backhaul.
        \begin{align}
        \mathbb{P}\left[R_{s,4}>\rho \right]=\mathbb{P}_{s,\mathrm{NL}}^{cov}(2^{\frac{\rho}{\eta}}-1|r_{s} ) \times  \mathbb{P}_{bh,\mathrm{NL}}^{cov}(2^{\frac{(1-p_h)\rho}{(1-\eta) }}-1|r_{bh} )
    \end{align}
\end{itemize}
where $\mathbb{P}_{s,L}^{cov}(\cdot)$ and $\mathbb{P}_{s,NL}^{cov}(\cdot)$ are the SINR distributions of the user covered by LoS SBS and NLoS SBS, respectively.  $\mathbb{P}_{m,L}^{cov}(\cdot)$ and $\mathbb{P}_{m,NL}^{cov}(\cdot)$  are the SINR distributions of SBS covered by LoS MBS and NLoS MBS, respectively  (in Proposition \ref{ProUserSINRCoverageProbability}).

Following the same logic,  for the user associated with the MBS, the spectral efficiency distribution  is divided into two cases:
\begin{itemize}
  \item  The access link between the user and the MBS is LoS:
    \begin{align}
        \mathbb{P}\left[R_{m,1}>\rho \right]
        =\mathbb{P}\left[\eta  \log _{2}(1+\operatorname{SINR}_{m}^{\mathrm{L}} (r))>\rho\right]
        =\mathbb{P}\left[\mathrm{SINR}_{m}^{\mathrm{L}}(r)\geq 2^{\frac{\rho}{\eta }}-1\right]
        =  \mathbb{P}_{m,L}^{cov}(2^{\frac{\rho}{\eta}}-1|r_{m} )
    \end{align}
  \item  The access link between the user and the MBS is NLoS:
 \begin{align}
        \mathbb{P}\left[R_{m,2}>\rho \right]
        =\mathbb{P}\left[\mathrm{SINR}_{m}^{\mathrm{NL}}(r)\geq 2^{\frac{\rho}{\eta }}-1\right]
        =  \mathbb{P}_{m,\mathrm{NL}}^{cov}(2^{\frac{\rho}{\eta}}-1|r_{m} )
 \end{align}
\end{itemize}

where $\mathbb{P}_{s,L}^{cov}(\cdot)$ and $\mathbb{P}_{s,NL}^{cov}(\cdot)$ are the SINR distributions of the user covered by LoS MBS  and NLoS MBS, respectively (in Proposition \ref{ProUserSINRCoverageProbability}).

\subsection{The proof of Proposition \ref{InterferenceASE}}\label{Appen5:Inteference}
when the density of SBS is higher (i.e., $\lambda_s$$\rightarrow$$\infty$), the SBS will become more closer to the user. All the transmission signal from SBS will  be transmitted in an LoS channel to the user (i.e., $\mathcal{P}_{L}(r)=1$). That means the NLoS transmission of SBS is neglected. Besides, the  original  $\mathcal{\overline{L}}_{I_{s,m}}^{\mathrm{L}}$ in the general ASE of SBS (\ref{SBSAreaSpectralEfficiency}) will be approximated as:
\begin{small}
\begin{align}\label{INTLoSInterferenceFromtheSBS}
   \mathcal{\overline{L}}_{I_{s,m}}^{\mathrm{L,int}}\left( \gamma r^{\alpha_{\mathrm{L}}} \right)\thickapprox\exp \left(-2 \pi  \lambda_{s} \int_{r}^{\infty}  \frac{u}{1+\frac{r^{-\alpha_{\mathrm{L}}}}{\gamma u^{-\alpha_{\mathrm{L}}}}} d u
   -2 \pi  \lambda_{m} \int_{\left(\frac{P_{m}^{tr} }{P_{s}^{tr} }\right)^{\frac{1}{\alpha^{\mathrm{L}}}} r }^{\infty} \frac{u}{1+\frac{P_{s}^{tr}B_{s}h_sA_L r^{-\alpha_{\mathrm{L}}}}{\gamma P_{m}^{tr}B_{m}h_m A_L u^{-\alpha_{\mathrm{L}}}}} d u \right)
\end{align}
\end{small}
And the  original  $\mathcal{\overline{L}}_{I_{s,m}}^{\mathrm{NL}}$ in the general ASE of SBS (\ref{SBSAreaSpectralEfficiency}) will be approximated to zero in the   interference-limited case.

With the above analysis, based on the general ASE of SBS in (\ref{SBSAreaSpectralEfficiency}), the ASE of SBS in the interference-limited case is obtained and  given in the Proposition \ref{InterferenceASE}.


\end{document}